%% file: main.tex
\documentclass[a4paper,envcountsame]{llncs}
\usepackage{ifthen}

\usepackage{color}
\usepackage{amsmath,amsfonts}
\usepackage{txfonts}
\usepackage{mathtools}
\usepackage{caption}

\usepackage{tikz}
\usetikzlibrary{automata,positioning}
\tikzstyle{max}=[shape=rectangle,draw,inner sep=0pt,minimum size=6mm,thick]
\tikzstyle{ran}=[shape=circle,draw,inner sep=0pt,minimum size=6mm,thick]
\tikzstyle{tran}=[thick,draw,->,>=stealth]
\newcommand{\Zset}{\mathbb{Z}}

\newcommand{\C}{\mathcal{C}}

\newcommand{\calO}{\mathcal{O}}
\newcommand{\Nset}{\mathbb{N}}

\newcommand{\eps}{\varepsilon}
\renewcommand{\rho}{\varrho}
\newcommand{\minec}{\textit{min-ec}}
\newcommand{\ratio}{\mathit{rat}}
\newcommand{\minpath}{\mathit{MinPath}}
\newcommand{\avoid}{\mathit{Avoid}}
\newcommand{\reach}{\mathit{Reach}}
\newcommand{\minpathr}{\mathit{MinPathReach}}
\newcommand{\cmax}{\cost_{\max}}
\newcommand{\size}[1]{||#1||}

\newcommand{\B}{\mathcal{B}}

\newcommand{\Ca}{\mathit{cap}}
\newcommand{\Val}{\mathit{Val}}
\newcommand{\cost}{\mathit{c}}
\newcommand{\finmem}{\mu}
\newcommand{\counting}{\kappa}
\newcommand{\run}[2]{\mathit{run}(#1,#2)}
\newcommand{\Run}{\mathit{Run}}
\newcommand{\mcs}{\mathit{MCS}}
\newcommand{\mcf}{\mathit{MCF}}
\newcommand{\T}{\textit{T}}

\newcommand{\genTran}[2]{%
    {}\mathchoice%
    {\stackrel{#1}{#2}}
    {\mathop {\smash{#2}}\limits^{\vrule width 0pt height 0pt depth 4pt\smash{#1}}}
    {\stackrel{#1}{#2}}
    {\stackrel{#1}{#2}}
{}}

\newcommand{\tran}[1]{\genTran{#1}{\rightarrow}}

\newcommand{\bbtran}[1]{\genTran{#1}{\mapsto}}

\newcommand{\len}[1]{len(#1)}
\newcommand{\NP}{\textbf{NP}}
\newcommand{\coNP}{\textbf{co-NP}}

\newcommand{\PTIME}{\textbf{P}}

\newcommand{\ifApp}[2]%
{\ifthenelse{\isundefined{\showappendix}}{#2}{#1}}

\newcommand{\MC}{\mathit{MC}}
\newcommand{\MP}{\mathit{MP}}

\newcommand{\Mem}{\mathit{Mem}}
\newcommand{\CAct}{\mathit{Act}}
\newcommand{\dec}{\mathit{dec}}
\newcommand{\reset}{\mathit{reset}}
\newcommand{\noc}{\mathit{noc}}

\newcommand{\cpath}{\mathit{CPath}}

\newcommand{\ctuple}{\mathit{CTuple}}
\newcommand{\rts}{$r$-$t$-segment }
\newcommand{\dc}{\mathit{dc}}
\newcommand{\rank}{\mathit{rank}}
\newcommand{\last}{\mathit{last}}
\newcommand{\minsafe}{\mathit{MS}}
\newcommand{\ec}{\mathit{ec}}
\newcommand{\plusca}{\oplus_{\Ca}}

\newcommand{\theoremlike}[2]{\par\medskip\penalty-250%
{{\bfseries\noindent
#2 \ref{#1}.}}\it}

\newcommand{\thmhelperpre}[2]{\theoremlike{#1}{#2}}
\newcommand{\thmhelperpost}{\par\medskip}

\newenvironment{reflemma}[1]{\thmhelperpre{#1}{Lemma}}{\thmhelperpost}

\title{Minimizing Running Costs in Consumption Systems}
 \author{
 Tom\'a\v{s} Br\'{a}zdil  \and
 David Kla\v{s}ka  \and
 Anton\'in Ku\v{c}era \and
 Petr Novotn\'y
 }
 \authorrunning{
 Br\'{a}zdil \and
 Kla\v{s}ka \and
 Ku\v{c}era \and
 Novotn\'{y}
 }

 \institute{
 Faculty of Informatics, Masaryk University, Brno, Czech Republic\\
 }
\begin{document}

\maketitle

\begin{abstract}
A standard approach to optimizing long-run running costs of discrete 
systems is based on minimizing the \emph{mean-payoff}, i.e., the long-run
average amount of resources (``energy'') consumed per transition. However, 
this approach inherently assumes that the energy source has an 
unbounded capacity, which is not always realistic. For example,
an autonomous robotic device has a battery of finite capacity that
has to be recharged periodically, and the total amount of energy consumed
between two successive charging cycles is bounded by the capacity.
Hence, a controller minimizing the mean-payoff must obey this 
restriction.
In this paper we study the controller synthesis problem for
\emph{consumption systems} with a finite battery capacity, where the
task of the controller is to minimize the mean-payoff while preserving
the functionality of the system encoded by a given linear-time
property. We show that an optimal controller always exists, and it may
either need only finite memory or require infinite memory (it
is decidable in polynomial time which of the two cases
holds). Further, we show how to compute an effective description of an
optimal controller in polynomial time.  Finally, we consider the limit
values achievable by larger and larger battery capacity, show that
these values are computable in polynomial time, and we also analyze
the corresponding rate of convergence.  To the best of our knowledge,
these are the first results about optimizing the long-run running
costs in systems with bounded energy stores.
\end{abstract}

\input{intro}

\input{defs}

\input{optimal}
\input{nobuchi}

\input{limit}
\input{concl}
%
%
%
%

\bibliography{str-long,concur}

\newpage
\appendix

\noindent
\begin{center}
\huge\bf Technical Appendix
\end{center}

\input{app-example}
\input{app-cap-reach-alg}
\input{app-algorithms}

\end{document}

%% file: intro.tex
\section{Introduction}
\label{sec-intro}

A standard tool for modelling and analyzing the long-run average
running costs in discrete systems is \emph{mean-payoff}, i.e., the
average amount of resources (or ``energy'') consumed per
transition. More precisely, a system is modeled as a finite directed
graph~$\C$, where the set of states $S$ corresponds to configurations,
and transitions model the discrete computational steps. Each
transition is labeled by a non-negative integer specifying the amount
of energy consumed by a given transition. 
Then, to every run $\varrho$ in~$\C$, one can assign the associated 
\emph{mean-payoff}, which is the limit of average energy 
consumption per transition computed for longer and longer prefixes 
of~$\varrho$. A basic algorithmic task is to find a suitable 
\emph{controller} for a given system which minimizes the mean-payoff.
Recently, the problem has been generalized by requiring that the
controller should also achieve a given \emph{linear time property}
$\varphi$, i.e., the run produced by a controller should satisfy
$\varphi$ while minimizing the mean-payoff (see, e.g., 
\cite{CHJ:MP-parity-games}).
This is motivated by the fact that the system is usually required to
achieve some functionality, and not just ``run'' with minimal 
average costs. 

Note that in the above approach, it is inherently assumed that all
transitions are always enabled, i.e., the amount of energy consumed by
a transition is always available. In this paper, we study the long-run
average running costs in systems where the energy stores 
(``tanks'' or  ``batteries'') 
have a \emph{finite} capacity~$\Ca \in \Nset$. As before, the energy 
stored in the battery is consumed by performing transitions, but if the
amount of energy currently stored in the battery is smaller than the
amount of energy required by a given transition, then the transition
is disabled. From time to time, the battery must be reloaded,
which is possible only in certain situations (e.g., when visiting 
a petrol station).
These restrictions are directly reflected in our model, where 
some states of $\C$ are declared as \emph{reload states}, and the
run produced by a controller must be \mbox{\emph{$\Ca$-bounded}}, i.e.,
the total amount of energy consumed between two successive visits to 
reload states does not exceed~$\Ca$.  
\smallskip
 
\noindent
\textbf{The main results} of this paper can be summarized as follows.
Let $\C$ be a system (with a given subset of reload states) and 
$\varphi$ a linear-time property encoded as a non-deterministic 
B\"{u}chi automaton.
\begin{itemize}
\item[(A)] 
  We show that for a given capacity $\Ca \in \Nset$ and a given state $s$ of 
  $\C$, there exists a  controller $\mu$ \emph{optimal} for $s$ which 
  produces a $\Ca$-bounded run satisfying 
  $\varphi$ while minimizing the mean payoff. Further, we prove that
  there is a dichotomy in the structural complexity of~$\mu$, i.e., one 
  of the following possibilities holds:
  \begin{itemize}
  \item The controller $\mu$ can be constructed so that it has 
     finitely many memory elements and
     can be compactly represented as a \emph{counting controller}
     $\kappa$ which is computable in time polynomial in the size of
     $\C$ and $\Ca$ (all integer constants are encoded in \emph{binary}).
  \item The controller $\mu$ \emph{requires} infinite memory (i.e., 
     every optimal controller has infinite memory) and there exists
     an optimal \emph{advancing controller} $\pi$ which admits a finite
     description computable in time
     polynomial in the size of $\C$ and $\Ca$.
  \end{itemize}
  Further, we show that it is decidable in polynomial time which of
  the two possibilities holds.
\item[(B)] For every state $s$ of $\C$, we consider its \emph{limit value},
  which is the \emph{inf} of all mean-payoffs achievable by controllers   
  for larger and larger battery capacity. We show that the limit value is
  computable in polynomial time. Further, we show that the problem 
  whether the limit value is achievable by some \emph{fixed} finite battery 
  capacity is decidable in polynomial time. If it is the case,
  we give an explicit upper bound for~$\Ca$; and if not, we give an upper bound
  for the difference between the limit value and the best mean-payoff
  achievable for a given capacity~$\Ca$. 
\end{itemize}

\noindent
Technically, the most difficult part is~(A), where we need to analyze
the structure of optimal controllers and invent some tricks that
allow for compact representation and computation of optimal controllers.
Note that all constants are encoded in binary, and hence we cannot afford
to construct any ``unfoldings'' of $\C$ where the current
battery status (i.e., an integer between $0$ and $\Ca$) is 
explicitly represented, because such an unfolding is exponentially
larger than the problem instance. This is overcome by non-trivial
insights into the structure of optimal controllers.
\smallskip

\noindent
\textbf{Previous and related work.} A combination of mean-payoff and
linear-time (parity) objectives has been first studied in 
\cite{CHJ:MP-parity-games} for
2-player games. It has been shown that optimal strategies exist
in such games, but they may require infinite memory. Further,
the values can be computed in time which is pseudo-polynomial in the size
of the game and exponential in the number of priorities. 
Another closely related formalisms are \emph{energy games}
and \emph{one-counter games}, where each transition can both
increase and decrease the amount of energy, and the basic task
of the controller is to avoid the situation when the battery 
is empty. Energy games with parity objectives have been considered
in \cite{CHD:energy-games}. In these games,
the controller also needs to satisfy a given parity condition
apart of avoiding zero. Polynomial-time algorithms for certain subclasses of
``pure'' energy games (with zero avoidance objective only) have recently
been designed in \cite{CHKN:energy-games-polynomial}. Energy games with capacity constraints were studied in \cite{FJLS:multi-energy-games}. Here it was shown, that deciding whether a given one-player energy game admits a run along which the accumulated reward stays between 0 and a given positive capacity is already an NP-hard problem. 
\emph{One-counter Markov decision processes} and \emph{one-counter
  stochastic games}, where the counter may change at most by one in
each transition, have been studied in \cite{BBEKW:OC-MDP,BBE:OC-games}
for the objective of \emph{zero reachability}, which is dual to zero
avoidance.  It has been shown that for one-counter MDPs (both
maximizing and minimizing), the existence of a controller that reaches
zero with probability one is in \PTIME. If such a controller exists,
it is computable in polynomial time. For one-counter stochastic games,
it was shown that the same problem is in $\NP\cap\coNP$. In
\cite{BKNW:OC-MDP-term-time}, it was shown how to compute an 
\mbox{$\varepsilon$-optimal} controller minimizing the expected number 
of transitions needed to visit zero in one-counter MDPs. 
Another related model with only one counter are \emph{energy Markov decision
  processes} \cite{CHD:energy-MDPs}, where the counter updates are
arbitrary integers encoded in binary, and the controller aims at
maximizing the probability of all runs that avoid visiting zero and
satisfy a given parity condition. The main result of
\cite{CHD:energy-MDPs} says that the existence of a controller such
that the probability of all runs satisfying the above condition is
equal to one for a sufficiently large initial counter value is in
$\NP\cap\coNP$. Yet another related model are \emph{solvency games}
\cite{BKSV:Solvency-games}, which can be seen as rather special
one-counter Markov decision processes (with counter updates encoded in
binary). The questions studied in \cite{BKSV:Solvency-games} concern
the structure of an optimal controller for maximizing the probability
of all runs that avoid visiting negative values, which is closely
related to zero avoidance. 

There are also results about systems with more than one counter 
(resource). Examples include games over vector addition systems with
states \cite{BJK:eVASS-games}, \emph{multiweighted energy games}
\cite{FJLS:multi-energy-games,BFLMS:weighted-automata-inf-runs}, 
\emph{generalized energy games} 
\cite{CHDHR:energy-mean-payoff}, \emph{consumption
games} \cite{BCKN:consumption-games}, etc. We refer to 
\cite{Kucera:multicounter-games} for a more detailed overview.

%% file: defs.tex
\section{Preliminaries}
\label{sec-defs}
The sets of all integers, positive integers, and non-negative integers
are denoted by $\Zset$, $\Nset$, and $\Nset_0$, respectively.  Given a
set $A$, we use $|A|$ to denote the cardinality of $A$.  The encoding
size of a given object $B$ is denoted by $\size B$.  In particular,
all integer numbers are encoded in \emph{binary}, unless otherwise
stated. The $i$-th component of a vector (or tuple) $v$ is denoted 
by $v[i]$.

A \emph{labelled graph} is a tuple $G = (V,\tran{},L,\ell)$ where $V$ is
a non-empty finite set of \emph{vertices}, ${\tran{}} \subseteq V \times V$ 
is a set of \emph{edges}, $L$ is a non-empty finite set of \emph{labels},
and $\ell$ is a function which to every edge assigns a label of~$L$.
We write $s\tran{a} t$ if $s\tran{} t$ and $a$ is the label of~$(s,t)$. 

A \emph{finite path} in $G$ of \emph{length}~$n \in \Nset_0$ is a
finite sequence $\alpha \equiv v_0 \ldots v_n$ of vertices such that
$v_i \tran{} v_{i+1}$ for all $0 \leq i < n$. The length of $\alpha$
is denoted by $\len{\alpha}$, and the label of $v_i \tran{} v_{i+1}$
is denoted by $a_i$.
An \emph{infinite path} (or \emph{run}) in $G$ is an infinite sequence
of vertices $\varrho$ such that every finite prefix of $\varrho$ is a
finite path in~$G$. Finite paths and runs in $G$ are also written as
sequences of the form $v_0 \tran{a_0} v_1 \tran{a_1} v_2 \tran{a_2}
\cdots$.  Given a finite or infinite path $\varrho \equiv v_0\, v_1
\ldots$ and $i\in \Nset_0$, we use $\varrho(i)$ to denote the $i$-th
vertex $v_i$ of $\varrho$, and $\varrho_{\leq i}$ to denote the prefix $v_0
\ldots v_i$ of $\varrho$ of length~$i$.

A finite path $\alpha \equiv v_0 \ldots v_n$ in $G$ is a \emph{cycle}
if $n\geq 1$ and $v_0 = v_n$, and a \emph{simple cycle} if it is a cycle and $v_{i}
\neq v_j$ for all $0\leq i < j < n$. Given a finite path $\alpha
\equiv v_0 \ldots v_n$ and a finite or infinite path $\varrho \equiv
u_0\,u_1\ldots$ such that $v_n = u_0$, we use $\alpha\cdot\varrho$ to
denote the \emph{concatenation} of $\alpha$ and $\rho$, i.e., the path
$v_0 \ldots v_n\,u_1 \, u_2\ldots$ Further, if $\alpha$ is a cycle, we
denote by $\alpha^{\omega}$ the infinite path
$\alpha\cdot\alpha\cdot\alpha\cdots$.

In our next definition, we introduce consumption systems that have
been informally described in Section~\ref{sec-intro}. Recall that
an optimal controller for a consumption system should minimize
the mean-payoff of a $\Ca$-bounded run and satisfy a given linear-time
property~$\varphi$ (encoded by a non-deterministic B\"{u}chi 
automaton $\B$). For technical convenience, we assume that
$\B$ has already been multiplied with the considered consumption 
system (i.e., the synchronized product has already been constructed\footnote{It will become clear later that $\B$ being non-deterministic is not an obstacle here, since we work in a non-stochastic one-player setting.}).
Technically, we declare some states in consumption systems as 
accepting and require that a \mbox{$\Ca$-bounded} run visits an accepting
state infinitely often. 

\begin{definition}
\label{def-CS}
A \emph{consumption system} is a tuple $\C = (S,\tran{},c,R,F)$ where $S$
is a finite non-empty set of \emph{states}, $\tran{} \subseteq S
\times S$ is a \emph{transition relation}, $c$ is a function
assigning a non-negative integer \emph{cost} to every transition,
$R\subseteq S$ is a set of \emph{reload states}, and $F\subseteq S$
a non-empty set of \emph{accepting states}.  
We assume that $\tran{}$ is \emph{total}, i.e.,
for every $s \in S$ there is some $t\in S$ such that $s \tran{} t$.
\end{definition}
The encoding size of $\C$ is denoted by $\size \C$ (transition costs
are encoded in binary). All notions defined for labelled graphs
naturally extend to consumption systems.

The \emph{total cost} of a given finite path $\alpha \equiv s_0\tran{c_0}
s_1\tran{c_1} \cdots \tran{c_n} s_{n+1}$ is defined as
$\cost(\alpha)=\sum_{i=0}^n c_i$, and the \emph{mean
  cost} of $\alpha$ as $\MC(\alpha)=c(\alpha)/(n{+}1)$.
Further, we define the \emph{end cost} of $\alpha$ as the total cost 
of the longest suffix $s_{i}\tran{c_{i}} \cdots \tran{c_n} s_{n+1}$ of 
$\alpha$ such that $s_{i+1},\ldots,s_{n+1}\not\in R$ (intuitively, the end cost
of $\alpha$ is the total amount of resources consumed since the last 
reload).  

Let $\Ca \in \Nset$. We say that a finite or
infinite path $\varrho \equiv s_0\tran{c_0} s_1\tran{c_1} s_2\tran{c_2}\cdots$
is \mbox{\emph{$\Ca$-bounded}} if the end cost of every finite prefix of
$\varrho$ is bounded by $\Ca$ (intuitively, this means that the total  amount
of resources consumed between two consecutive visits to reload states
in $\varrho$ is bounded by $\Ca$). Further, we say a run $\varrho$ in $\C$
is \emph{accepting} if $\varrho(i) \in F$ for infinitely many $i \in \Nset$.
For every run $\varrho$ in $\C$
we define
\[
   \Val^{\Ca}_{\C}(\varrho) = 
     \begin{cases}
       \limsup_{i \rightarrow \infty} \MC(\varrho_{\leq i})%
                & \mbox{if $\rho$ is $\Ca$-bounded and accepting};\\
       \infty   & \mbox{otherwise.}
     \end{cases}
\]
The \emph{$\Ca$-value} of a given state $s \in S$ is defined by
\[
\Val^{\Ca}_{\C}(s) = \inf_{\varrho\in \Run(s)} \Val^{\Ca}_{\C}(\varrho)
\]
where $\Run(s)$ is the set of all runs in $\C$ initiated in~$s$.
Intuitively, $\Val^{\Ca}_{\C}(s)$ is the minimal mean cost of a
$\Ca$-bounded accepting run initiated in~$s$. The \emph{limit value}
of~$s$ is defined by $\Val_{\C}(s)=\lim_{\Ca\rightarrow \infty} \Val^{\Ca}_{\C}(s)$.

\begin{definition}
Let $\C = (S,\tran{},c,R,F)$ be a consumption system. 
A \emph{controller} for $\C$ is a tuple $\finmem=(M,\sigma_n,\sigma_u,m_0)$
where $M$ is a set of \emph{memory elements}, $\sigma_n : S
\times M\rightarrow S$ is a \emph{next function} satisfying $s\tran{}
\sigma_n(s,m)$ for every $(s,m)\in S\times M$, $\sigma_u: S \times
M \rightarrow M$ is an \emph{update function}, and $m_0$ is an 
\emph{initial memory element}. If $M$ is finite, we say that
$\finmem$ is a \emph{finite-memory} controller (FMC).
\end{definition}

\noindent
For every finite path $\alpha=s_0\dots s_n$ in $\C$, we use
$\hat{\sigma}_u(\alpha)$ to denote the unique memory element
``entered'' by $\mu$ after reading~$\alpha$. Formally,
$\hat{\sigma}_u(\alpha)$ is defined inductively by
$\hat{\sigma}_u(s_0)=\sigma_u(s_0,m_0)$, and $\hat{\sigma}_u(s_0
\ldots s_{n+1}) = \sigma_u(s_{n+1},\hat{\sigma}_u(s_0 \ldots
s_n))$. Observe that for every $s_0 \in S$, the controller $\mu$
determines a unique run $\run{\finmem}{s_0}$ defined as follows: the
initial state of $\run{\finmem}{s_0}$ is $s_0$, and if $s_0 \ldots
s_n$ is a prefix of $\run{\finmem}{s_0}$, then the next state is
$\sigma_n(s_n,\hat{\sigma}_u(s_0 \ldots s_n))$.  
The size of a given FMC $\mu$ is denoted by $\size{\mu}$ (in particular,
note that $\size{\mu} \geq |M|$).

\begin{definition}
Let $\C$ be a consumption system, $\mu$ a controller for $\C$, and
$\Ca \in \Nset$. We say that $\mu$ is \emph{$\Ca$-optimal} for 
a given state $s$ of $\C$ if 
$\Val^{\Ca}_{\C}(\run{\finmem}{s})  =   \Val^{\Ca}_{\C}(s)$. 
\end{definition}

\noindent
As we shall see, an optimal controller for~$s$ always exists, but it may
require infinite memory. Further, even if there is a FMC for~$s$, it
may require exponentially many memory elements. To see this, consider
the simple consumption system of Fig.~\ref{fig-exp-mem}. An optimal controller
for $s$ has to (repeatedly) perform $\Ca - 10$ visits to $t$ and then one 
visit to the only reload state $u$, which requires $\Ca - 10$ memory 
elements (recall that $\Ca$ is encoded in binary). Further examples of a non-trivial optimal behaviour can be found in Appendix~\ref{app-example}.

To overcome these difficulties, we introduce a special type 
of finite-memory controllers called \emph{counting controllers},
and a special type of infinite memory controllers called
\emph{advancing controllers}.

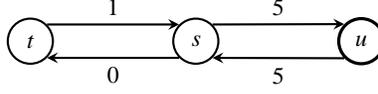
\begin{figure}[t]
\begin{center}
\begin{tikzpicture}[x=2.2cm,y=0cm,font=\small]
\node (t) at (0,0) [ran] {$t$};
\node (s) at (1,0) [ran] {$s$};
\node (u) at (2,0) [ran, very thick] {$u$};
\draw [tran] (t.45) -- node[above] {$1$} (s.135);
\draw [tran] (s.45) -- node[above] {$5$} (u.135);
\draw [tran] (s.225) -- node[below] {$0$} (t.315);
\draw [tran] (u.225) -- node[below] {$5$} (s.315);
\end{tikzpicture}
\end{center}
\caption{An optimal controller may require memory of exponential size. Here $R=\{u\}$ and $F=S$.}
\label{fig-exp-mem}
\end{figure}

Intuitively, memory elements of a counting controller are pairs
of the form $(r,d)$ where $r$ ranges over a finite set $\Mem$
and $d$ is a non-negative integer of a bounded size. The next
and update functions depend only on $r$ and the information whether
$d$ is zero or positive. The update function may change $(r,d)$ to some
$(r',d')$ where $d'$ is obtained from $d$ by performing a \emph{counter
action}, i.e., an instruction of the form $\dec$ (decrement), 
$\noc$ (no change), or $\reset(n)$ where $n \in \Nset$ (reset the 
value to $n$). Hence, counting controllers admit a compact representation
which utilizes the special structure of memory elements and the mentioned
restrictions.

\begin{definition}
Let $\C = (S,\tran{},c,R,F)$ be a consumption system. A 
\emph{counting controller} for $\C$ is a tuple 
$\counting=(\Mem,\sigma^+_n,\sigma^0_n,\CAct,\sigma^+_u,\sigma^0_u,r_0)$
where
\begin{itemize}
\item $\Mem$ is a finite set of \emph{basic memory elements},
\item $\sigma^+_n,\sigma^0_n: S \times \Mem \rightarrow S$ are 
   \emph{positive} and \emph{zero next functions} satisfying 
   $s \tran{} \sigma^+_n(s,r)$ and $s \tran{} \sigma^0_n(s,r)$ for 
   every $(s,r) \in S \times \Mem$, respectively,
\item $\CAct$ is a finite set of \emph{counter actions} (note that
   $\CAct$ may contain instructions of the form $\reset(n)$
   for different constants~$n$);
 \item $\sigma^{+}_u : S \times \Mem \rightarrow
   \Mem \times \CAct$ is a \emph{positive update function},
 \item $\sigma^{0}_u:S\times \mathit{Mem}\rightarrow
   \mathit{Mem}\times (\CAct\smallsetminus\{\dec\})$ is a \emph{zero update
     function},
\item $r_0\in \Mem$ is an initial basic memory element.
\end{itemize}
\end{definition}
\noindent
The encoding size of a counting controller $\kappa$ is denoted by 
$\size\kappa$, where all constants used in counter actions are 
encoded in binary.

The functionality of a counting controller 
$\kappa= (\mathit{Mem},\sigma^+_n,\sigma^0_n,\CAct,\sigma^+_u,\sigma^0_u,r_0)$
is determined by its associated  finite-memory controller 
$\mu_\kappa = (M,\sigma_n,\sigma_u,m_0)$ where
\begin{itemize}
\item $M=\Mem \times \{0,\ldots,k_{max}\}$ where 
  $k_{max}$ is the largest $n$ such that $\reset(n)\in \CAct$ (or $0$ if 
  no such $n$ exists); 
\item $\sigma_n(s,(r,d)) = \sigma_n^\odot(s,r)$, where $\odot$ is either
  $+$ or $0$ depending on whether $d>0$ or $d=0$, respectively;
\item $\sigma_u(s,(r,d)) = (r',d')$, where 
  $r'$ is the first component of $\sigma_u^\odot(s,r)$, and
  $d'$ is either $d$, $d-1$, or $n$, depending on whether the
  counter action in the second component of $\sigma_u^\odot(s,r)$
  is $\noc$, $\dec$, or $\reset(n)$, respectively (again,
  $\odot$ is either $+$ or $0$ depending on whether $d>0$ or $d=0$);
\item $m_0=(r_0,0)$.
\end{itemize}
Observe that $\size\kappa$ can be exponentially smaller than
$\size{\mu_\kappa}$. Slightly abusing our notation, we write
$\run{\counting}{s_0}$ instead of $\run{\mu_{\counting}}{s_0}$.

A counting controller $\kappa$ can be seen as a program for a
computational device with $\calO(\size{\Mem})$ control states and
$\log(k_{max})$ bits of memory needed to represent the bounded
counter. This device ``implements'' the functionality of
$\mu_\kappa$.

\begin{definition}
  Let $\C = (S,\tran{},c,R,F)$ be a consumption system and $s \in S$. 
  An \emph{advancing controller} for $\C$ and $s$ is a controller 
  $\pi$ for $\C$ such that $\run{\pi}{s}$ takes the form
  $\alpha \cdot \beta \cdot \gamma \cdot \beta^2
   \cdot \gamma \cdot \beta^4 \cdots \gamma \cdot \beta^{2^i} \cdots$ 
  where $\beta(0) \neq \beta(i)$ for all $0 < i < \len{\beta}$.
\end{definition}
The encoding size of an advancing controller $\pi$, denoted by 
$\size{\pi}$, is given by the total encoding size 
of $\alpha$, $\beta$, and $\gamma$. 
Typically, $\alpha$ and $\gamma$ will be of polynomial length,
but the length of $\beta$ is sometimes exponential and 
in this case we use a counting controller to represent 
$\beta$ compactly. Formally, we say that \emph{$\size{\pi}$ is polynomial}
in $\size{\C}$ and $\size{\Ca}$ if $\alpha$ and $\gamma$
are of polynomial length and there exists a counting controller
$\kappa[\beta]$ such that $\run{\kappa[\beta]}{\beta(0)} = \beta^\omega$
and $\size{\kappa}$ is polynomial in $\size{\C}$ and $\size{\Ca}$.

An advancing controller $\pi$ can be seen as a program for a
computational device equipped with two unbounded counters 
(the first counter maintains the current $i$ and the other counter is used
to count from $2^i$ to zero; if the device cannot implement the `$2^x$'
function directly, an auxiliary third counter may be needed). Also
note that the device can use the program of $\kappa[\beta]$ as a subroutine
to produce the finite path $\beta$ (and hence also finite paths of the form
$\beta^{2^i}$). Since $\beta(0) \neq \beta(i)$ for all $0 < i < \len{\beta}$,
the device simply simulates $\kappa[\beta]$ until revisiting 
$\beta(0)$.

%% file: optimal.tex
\section{The Results}
\label{sec:results}

\noindent
In this section, we present the main results of our paper. Our first
theorem concerns the existence and computability of values and 
optimal controllers in consumption systems.

\vspace{-0.2cm}
\begin{theorem}
\label{thm:alg}
Let $\C$ be a consumption system, $\Ca \in \Nset$, and $s$ a state of~$\C$.
Then $\Val_\C^{\Ca}(s)$ is computable in polynomial time (i.e., 
in time polynomial in $\size\C$ and $\size{\Ca}$, where $\Ca$ is 
encoded in binary). Further, there exists an optimal controller 
for~$s$. The existence of an optimal \emph{finite memory} controller
for~$s$ is decidable in polynomial time. If there exists an optimal
FMC for $s$, then there also exists an optimal \emph{counting} controller 
for~$s$ computable in polynomial time. Otherwise, there exists an
optimal \emph{advancing} controller for~$s$ computable in polynomial time.
\end{theorem}

\noindent
Our second theorem concerns the limit values, achievability of limit values,
and the rate of convergence to limit values.

\vspace{-0.2cm}
\begin{theorem}
\label{thm:limit}
Let $\C$ be a consumption system, $\Ca \in \Nset$, and $s$ a state of~$\C$.
Then $\Val_\C(s)$ can be computed in polynomial time (i.e.,
in time polynomial in $\size\C$).

Further, the problem whether $\Val_{\C}(s)=\Val_{\C}^{\Ca}(s)$ for 
some sufficiently large $\Ca\in \Nset$ is decidable in polynomial time.
If the answer is positive, then 
$\Val_{\C}(s) = \Val_{\C}^{\Ca}(s)$ for every $\Ca\geq 3\cdot|S|\cdot\cmax$, where $\cmax$ is the maximal cost of a transition in $\C$. 
Otherwise, for every $ \Ca> 4\cdot|S|\cdot\cmax$ 
we have that
$
\Val_{\C}^{\Ca}(s)-\Val_{\C}(s)\leq ({3\cdot|S|\cdot\cmax})/({\Ca-4\cdot|S|\cdot \cmax}).
$

\end{theorem}

\noindent
The next subsections are devoted to the proofs of 
Theorems~\ref{thm:alg}~and~\ref{thm:limit}. Due to space constrains,
some proofs and algorithms have been shifted to Technical Appendix. 

\vspace{-0.3cm}
\subsection{A Proof of Theorem~\ref{thm:alg}}
\label{sec:main-thm-proof}

For the rest of this section, we fix a consumption system
$\C=(S,\rightarrow,c,R,F)$, a capacity $\Ca \in \Nset$, and an initial
state $s \in S$. 

An \emph{admissibility witness} for a state $q \in S$ is
a cycle $\gamma$ initiated in $q$ such that $\gamma$ contains an accepting
state and there is a $\Ca$-bounded run initiated in $s$ of the 
form $\alpha \cdot \gamma^\omega$. We say that $q \in S$
is \emph{admissible} if there is at least one admissibility witness
for~$q$.

Observe that if $\gamma$ is an admissibility witness for 
a reload state~$q$, then $\gamma$ can be freely ``inserted'' into 
any $\Ca$-bounded run of the
form $\xi \cdot \delta$ where $\delta(0) = q$ so that the run $\xi
\cdot \gamma \cdot \delta$ is again $\Ca$-bounded. Such
simple observations about admissibility witnesses are frequently 
used in our proof of Theorem~\ref{thm:alg}, which is obtained 
in several steps:

\begin{itemize}
\item[(1)] We show how to compute all states $t \in S$ such that
  $\Val_\C^\Ca(t)=\infty$.
  Note that if  $\Val_\C^\Ca(t)=\infty$, then \emph{every} 
  controller is optimal in~$t$. Hence, if $\Val_\C^\Ca(s)=\infty$,
  we are done. Otherwise, we remove all states with infinite value
  from $\C$ together with their adjacent transitions.
\item[(2)] We compute and remove all states $t \in S$
  that are not reachable from $s$ via a \mbox{$\Ca$-bounded} finite path.
  This ``cleaning'' procedure simplifies our considerations and 
  it can be performed in polynomial time.
\item[(3)] We show that  $\Val_\C^\Ca(s)=0$ iff 
  $\C$ contains a \emph{simple} cycle with zero total cost initiated
  in an admissible state 
  (such a cycle is called a \emph{zero-cost} cycle).
  Next, we show that if there is a zero-cost cycle $\beta$ containing 
  an accepting state, then there is an optimal FMC $\mu$ for $s$ of  
  \emph{polynomial} size such that $\run{\mu}{s} = \alpha\cdot\beta^\omega$. 
  Otherwise, \emph{every} optimal controller for~$s$ has infinite memory, 
  and we show how to compute finite paths $\alpha,\gamma$ of polynomial 
  length such that the ($\Ca$-bounded) run 
  \mbox{$\varrho \equiv \alpha\cdot\beta\cdot\gamma\cdot\beta^2\cdot\gamma
     \cdot\beta^4\cdots \gamma\cdot\beta^{2^i}\cdots$} initiated in $s$
  satisfies $\Val_\C^\Ca(\varrho)=0$. Thus, the finite paths
  $\alpha$, $\beta$, and $\gamma$ represent an optimal advancing
  controller of polynomial size. 

  The existence of a zero-cost cycle (and the existence of a 
  zero-cost cycle that contains an accepting state) is decidable
  in polynomial time. If a zero-cost cycle exists, we are done.
  Otherwise, we proceed to the next step.
\item[(4)] Now we assume that $\C$ does not contain a zero-cost cycle. We 
  show that there exist 
  \begin{itemize}
  \item a $\Ca$-bounded cycle $\beta$ initiated in an admissible state
    such that $\MC(\beta) \leq
    \MC(\delta)$ for every $\Ca$-bounded cycle~$\delta$ initiated
    in an admissible state, and $\beta(0)
    \neq \beta(i)$ for all $0 < i < \len{\beta}$;
   \item a $\Ca$-bounded cycle $\hat{\beta}$ containing an accepting
     state such that $\MC(\hat{\beta}) \leq \MC(\hat{\delta})$ for
     every $\Ca$-bounded cycle $\hat{\delta}$ containing an accepting
     state.
  \end{itemize}
  We prove that $\Val_\C^\Ca(s) = \MC(\beta)$. Further, we show the
  following:
  \begin{itemize}
  \item If $\MC(\beta) = \MC(\hat{\beta})$, then there exists an optimal
     FMC $\mu$ for $s$ such that $\run{\mu}{s} = \alpha\cdot\hat{\beta}^\omega$,
     where $\alpha$ is a finite path of polynomial length.
     In general, $\len{\hat{\beta}}$ (and hence also $\size{\mu}$) 
     is \emph{exponential} in $\size{\C}$ and $\size{\Ca}$.
     We show how to compute a \emph{counting} controller $\kappa[\hat{\beta}]$ 
     of \emph{polynomial} size such that 
     $\run{\kappa[\hat{\beta}]}{\hat{\beta}(0)} = \hat{\beta}^\omega$. 
     Since $\alpha$ is a finite path of polynomial length,
     we also obtain a counting controller $\kappa$ of polynomial 
     size such that $\run{\kappa}{s} = \run{\mu}{s}$.
  \item If $\MC(\beta) < \MC(\hat{\beta})$, then \emph{every} optimal 
     controller for $s$ has infinite memory, and we show how to 
     efficiently compute finite paths $\alpha,\gamma$ of polynomial 
     length and a counting controller $\kappa[\beta]$ of polynomial size
     such that $\run{\kappa[\beta]}{\beta(0)} = \beta^\omega$ and
     the run 
     $\varrho \equiv \alpha\cdot\beta\cdot\gamma\cdot\beta^2\cdot\gamma
     \cdot\beta^4\cdots \gamma\cdot\beta^{2^i}\cdots$ initiated in $s$
     satisfies $\Val_\C^\Ca(\varrho)=\Val_\C^\Ca(s)$. Thus, we obtain
     an optimal advancing controller $\pi$ for~$s$ of polynomial size.
  \end{itemize}
\end{itemize}

\noindent
We start with step~(1).
\begin{lemma}
\label{lem:infinite-value}
  Let $t \in S$. The problem whether $\Val_\C^\Ca(t)=\infty$ is
  decidable in polynomial time.
\end{lemma}

\noindent
The next lemma implements step~(2).

\begin{lemma}
\label{lem-cap-bounded-path}
  Let $t \in S$. The existence of a $\Ca$-bounded path from
  $s$~to~$t$ is decidable in polynomial time. Further, an example
  of a $\Ca$-bounded path from $s$~to~$t$ (if it exists) is computable
  in polynomial time.
\end{lemma}

\noindent
We also need the following lemma which says that for every admissible
state, there is an efficiently computable admissibility witness.

\begin{lemma}
\label{lem:algorithm-admissible}
The problem whether a given $q \in S$ is admissible is decidable
in polynomial time. Further, if $q$ is admissible, then 
there are finite paths $\alpha,\gamma$ computable in polynomial time
such that $\alpha \cdot \gamma^\omega$ is a $\Ca$-bounded run initiated 
in $s$ and $\gamma$ is an admissibility witness for $q$ of length
at most $6\cdot|S|^2$.
\end{lemma}

\noindent
As we already indicated in the description of step~(2), from now
on we assume that all states of $\C$ have a finite value and are
reachable from $s$ via a $\Ca$-bounded finite path. 
Recall that a \emph{zero-cost} cycle is a cycle in $\C$ initiated
in an admissible state with zero total cost.
Now we proceed to step~(3).

\begin{lemma}
\label{prop:opt-runs-zero char}
We have that $\Val_\C^\Ca(s)=0$ iff there exists a zero-cost cycle.
Further, the following holds:
\begin{enumerate}
\item If there is a zero-cost cycle $\beta$ containing an accepting
  state, then the run $\varrho \equiv \alpha \cdot \beta^\omega$,
  where $\alpha$ is a \mbox{$\Ca$-bounded} finite path
  from $s$ to $\beta(0)$, satisfies $\Val_\C^\Ca(\varrho) = \Val_\C^\Ca(s)$.
  Hence, there is a FMC $\mu$ optimal for $s$ where $\size{\mu}$
  is polynomial in $\size{\C}$ and $\size{\Ca}$.
\item If there is a zero-cost cycle $\beta$ but no zero-cost cycle
  contains an accepting state, then every \mbox{$\Ca$-optimal}
  controller for $s$ has infinite memory. Further, for a given
  zero-cost cycle $\beta$ there exist
  finite paths $\alpha$ and $\gamma$ computable in polynomial time
  such that the run $\varrho \equiv 
  \alpha\cdot\beta\cdot\gamma\cdot\beta^2\cdots \gamma\cdot\beta^{2^i}\cdots$
  satisfies $\Val_\C^\Ca(\varrho) = \Val_\C^\Ca(s)$. 
  Hence, there exist an advancing controller $\pi$ optimal for $s$ 
  where $\size{\pi}$ is polynomial in $\size{\C}$ and $\size{\Ca}$.
\end{enumerate}
\end{lemma}
\begin{proof}
If $\Val_\C^\Ca(s)=0$, there is an accepting run $\varrho$ initiated in
$s$ such that $\Val_\C^\Ca(\varrho)<1/|S|$. Let $\varrho'$ be an infinite
suffix of $\varrho$ such that all states that appear in $\varrho'$
appear infinitely often in $\varrho'$. This means that all states 
that appear in $\varrho'$ are admissible. Obviously, there is
$k \in \Nset$ such that the cost of every transition 
$\varrho'(k{+}i) \tran{} \varrho'(k{+}i{+}1)$, where $0\leq i \leq |S|-1$, is
zero (otherwise, we would have $\Val_\C^\Ca(\varrho)=\Val_\C^\Ca(\varrho')\geq 1/|S|$), and hence there exists a zero-cost cycle. 

Now assume that $\C$ contains a zero-cost cycle $\beta$ containing an 
accepting state. Since there is a $\Ca$-bounded finite path $\alpha$
form $s$ to $\beta(0)$ (see step~(2) and Lemma~\ref{lem-cap-bounded-path}),
the run $\varrho \equiv \alpha \cdot \beta^\omega$ is $\Ca$-bounded and 
satisfies $\Val_\C^\Ca(\varrho)=0$. Since the length of $\alpha$ and $\beta$
is polynomial in  $\size{\C}$ and $\size{\Ca}$ (see 
Lemma~\ref{lem-cap-bounded-path}), we obtain Claim~1.

Finally, assume that $\C$ contains a zero-cost cycle $\beta$ but 
no zero-cost cycle in $\C$ contains an accepting state. Since $\beta(0)$ 
is admissible, there is a $\Ca$-bounded run $\alpha \cdot \gamma^\omega$
initiated in $s$ where $\gamma$ is an admissibility witness for 
$\beta(0)$. Note that the length of $\alpha$ and $\gamma$ is polynomial
in $\size{\C}$ and $\size{\Ca}$ by Lemma~\ref{lem:algorithm-admissible},
and the run $\varrho \equiv 
  \alpha\cdot\beta\cdot\gamma\cdot\beta^2\cdots \gamma\cdot\beta^{2^i}\cdots$
is accepting and $\Ca$-bounded. Further, a simple computation shows that
$\Val_\C^\Ca(\varrho)=0$. Hence, there exists an advancing controller 
$\pi$ optimal for $s$ such that $\size{\pi}$ is polynomial in $\size{\C}$ 
and $\size{\Ca}$. It remains to show that there is no optimal 
finite memory controller for~$s$. However, it suffices to realize that
if $\mu$ is a finite memory controller, then $\run{\mu}{s}$ takes the
form $\hat{\alpha}\cdot\hat{\beta}^\omega$, where $\hat{\beta}$ contains an
accepting state. By our assumption, $c(\hat{\beta}) \neq 0$, which means
that $\Val_\C^\Ca(\hat{\alpha}\cdot\hat{\beta}^\omega) \neq 0$. 
\qed
\end{proof}

\noindent
In the next lemma we show how to decide the existence of a zero-cost
cycle efficiently, and how to construct an example of a zero-cost cycle
if it exists. The same is achieved for zero-cost cycles containing
an accepting state. Thus, we finish step~(3).

\begin{lemma}
\label{lem-cycle-compute}
  The existence of a zero-cost cycle is decidable in polynomial  
  time, and an example of a zero-cost cycle $\beta$ (if it exists)
  is computable in polynomial time. The same holds for zero-cost cycles
  containing an accepting state.
\end{lemma}

\noindent
It remains to complete step~(4), which is the most technical part 
of our proof. From now on we assume that $\C$ does not contain any
zero-cost cycles. 

We say that a cycle $\beta$ in $\C$ is \emph{reload-simple}, if every 
reload state appears at most once in $\beta$, i.e., for every 
$t \in R$ there is at most one $0 \leq i < \len{\beta}$ satisfying
$\beta(i) = t$.
Further, we say that a cycle $\beta$ is \emph{\T-visiting}, where
$T\subseteq S$, if $\beta$ is a $\Ca$-bounded reload-simple cycle
initiated in an admissible reload state such that $\beta$ contains 
a state of~$T$. We say that $\beta$ is an \emph{optimal \T-visiting cycle}
if $\MC(\beta) \leq \MC(\delta)$ for every \T-visiting cycle~$\delta$.
Note that every state of a \T-visiting cycle $\beta$ is admissible.

\begin{lemma}
\label{lem:optimal-cycles-ex}
If $\C$ does not contain any zero-cost cycle, then it contains an
optimal \mbox{$F$-visiting} cycle and an optimal \mbox{$S$-visiting} cycle.
\end{lemma}
\begin{proof}
  We give an explicit proof just for $F$-visiting cycles (the argument
  for $S$-visiting cycles is very similar). First, we show that there is at
  least one $F$-visiting cycle, and then we prove that every
  $F$-visiting cycle has a bounded length. Thus, the set of all
  $F$-visiting cycles is finite, which implies the existence of an
  optimal one.

  Since $\Val_\C^\Ca(s)<\infty$, there is a $\Ca$-bounded accepting run 
  $\varrho$ initiated in $s$. Note that if $\varrho$ contained only finitely
  many occurrences of reload states, it would have to contain 
  zero-cost cycle, which contradicts our assumption. Hence, $\varrho$
  contains infinitely many occurrences of a reload state and infinitely
  many occurrences of an accepting state. Let $\varrho'$ be a suffix
  of $\varrho$ such that every state that appears in $\varrho'$
  appears infinitely often in $\varrho'$ (hence, all states
  that appear in $\varrho'$ are admissible). We say that  
  a subpath $\varrho'(i) \ldots \varrho'(j)$ of $\varrho'$ is 
  \emph{useless} if
  $\varrho'(i) = \varrho'(j) \in R$ and no accepting state is visited
  along this subpath. Let $\hat{\varrho}$ be a run obtained from
  $\varrho'$ by removing all useless subpaths (observe that $\hat{\varrho}$
  is still a $\Ca$-bounded accepting run). Then, there must 
  be a subpath $\hat{\varrho}(i) \ldots \hat{\varrho}(j)$ of $\hat{\varrho}$
  such that the length of this subpath is positive, 
  $\hat{\varrho}(i) = \hat{\varrho}(j) \in R$, the subpath
  visits an accepting state, and no reload state is visited more than
  once along $\hat{\varrho}(i) \ldots \hat{\varrho}(j{-}1)$. 
  Hence, this subpath is an $F$-visiting cycle.

  Now let $\beta$ be an $F$-visiting cycle. Then every state on
  $\beta$ is admissible, which means that every
  simple cycle $\delta$ that is a subpath of $\beta$ has positive
  cost, otherwise $\delta$ would be a zero-cost cycle. This implies 
  that a maximal length of a subpath of $\beta$ which
  does not contain any reload state is $(|S|+1)\cdot(\Ca+1)$ (because
  $\beta$ is $\Ca$-bounded). From the reload-simplicity of $\beta$ we get
  that $\len{\beta}\leq |R|\cdot(|S|+1)\cdot(\Ca+1)$.  
\qed
\end{proof}

\noindent
We use $\mcf$ and $\mcs$ to denote the mean cost of an optimal
$F$-visiting cycle and the mean cost of an optimal $S$-visiting
cycle, respectively. Now we prove the following:

\begin{lemma}
\label{prop:opt-runs-char}
Suppose that $\C$ does not contain any zero-cost cycle. Then 
$\Val_{\C}^{\Ca}(s) = \mcs \leq \mcf$. Moreover, the following holds:
\begin{enumerate}
\item If $\mcf=\mcs$, then for every optimal $F$-visiting cycle
  $\beta$ and every $\Ca$-bounded path $\alpha$ from $s$ to $\beta(0)$
  we have that the run $\varrho \equiv \alpha \cdot \beta^\omega$ satisfies
  $\Val_\C^\Ca(\varrho) = \Val_\C^\Ca(s)$.  Hence, there exists an
  optimal FMC  for $s$.
\item If $\mcs<\mcf$, then every $\Ca$-optimal controller for~$s$ 
  has infinite memory. Further, for a given optimal $S$-visiting cycle 
  $\beta$ there exist finite paths $\alpha$ and $\gamma$ computable 
  in polynomial time such that the run
  $\varrho \equiv 
  \alpha\cdot\beta\cdot\gamma\cdot\beta^2\cdots \gamma\cdot\beta^{2^i}\cdots$
  satisfies $\Val_\C^\Ca(\varrho) = \Val_\C^\Ca(s)$. 
  Hence, there exists an optimal advancing controller for $s$.
\end{enumerate}
\end{lemma}
\begin{proof}
Clearly, $\mcs\leq \mcf$, because every $F$-visiting cycle is also
$S$-visiting. Now  we show that for every run $\varrho$ we have that
$\Val_{\C}^\Ca(\varrho) \geq \mcs$. This clearly holds for all
non-accepting runs. Every accepting run $\varrho$ must contain
infinitely many occurrences of a reload state, otherwise it would
contain a zero-cost cycle as a subpath,  which contradicts our
assumption. Let $\varrho'$ be a suffix of $\varrho$ initiated in a reload 
state such that every state which appears in $\varrho'$ appears 
infinitely often in $\varrho'$. Then $\varrho'$ takes the form
$\beta_0\cdot\beta_1\cdot\beta_2\cdots$, where
for every $i \geq 0$, the subpath $\beta_i$ is a cycle initiated in a
reload state. Every $\beta_i$ can be decomposed into reload-simple cycles
$\beta_{i,1},\beta_{i,2},\dots,\beta_{i,i_m}$ that are initiated in
reload states (here the decomposition is meant in a graph-theoretical
sense, i.e., a transition appears $b$ times on $\beta_i$ if and only
if $b=b_1 + \cdots + b_m$, where $b_j$ is a number of occurrences of
this transition on $\beta_{i,j}$). Each of these cycles is an
$S$-visiting cycle (since every state on $\rho'$ is admissible) and
clearly $\MC(\rho)=\MC(\rho')\geq \min_{i\geq 1} MC(\beta_i) \geq
\min_{i\geq 0,1\leq j \leq i_m} \MC(\beta_{i,j}) \geq \mcs$.

The rest of the proof closely follows the proof of
Lemma~\ref{prop:opt-runs-zero char}. First we consider the case
when $\mcf=\mcs$, i.e., for every optimal $F$-visiting cycle $\beta$
we have that $\MC(\beta) = \mcs$. If $\alpha$ is a
$\Ca$-bounded path from $s$ to $\beta(0)$, then
we have that the run $\varrho \equiv \alpha \cdot \beta^\omega$ satisfies
$\Val_\C^\Ca(\alpha \cdot \beta^\omega) = \mcs = \Val_\C^\Ca(s)$, and
hence there exists an optimal FMC for~$s$. 

If $\mcs < \mcf$, consider an optimal $S$-visiting cycle $\beta$.
Since $\beta(0)$ is admissible, there is a $\Ca$-bounded run $\alpha
\cdot \gamma^\omega$ initiated in $s$ where $\gamma$ is an
admissibility witness for $\beta(0)$ and $\alpha$ and $\gamma$ are
computable in polynomial time (see
Lemma~\ref{lem:algorithm-admissible}. Further, the run $\varrho \equiv
\alpha\cdot\beta\cdot\gamma\cdot\beta^2\cdots
\gamma\cdot\beta^{2^i}\cdots$ is accepting and $\Ca$-bounded, and one
can easily show that $\Val_\C^\Ca(\varrho)= \mcs =
\Val_\C^\Ca(s)$. Hence, there exists an optimal advancing controller
for $s$. Since every finite memory controller $\mu$ satisfies 
$\run{\mu}{s} \equiv \hat{\alpha}\cdot\hat{\beta}^\omega$  
and $\size{\Ca}$. It remains to show that there is no optimal 
finite memory controller for~$s$. For every FMC $\mu$ we have that
$\run{\mu}{s} \equiv \hat{\alpha}\cdot\hat{\beta}^\omega$, where 
$\hat{\beta}$ is a cycle on a reload state containing an accepting 
state. Further, $\Val_\C^\Ca(\mu) = \MC(\hat{\beta})$. The cycle
$\hat{\beta}$ can be 
decomposed into reload-simple cycles on reloading 
states whose mean cost is at least $\mcs$. Since at least one of these
cycles is accepting and $\mcf > \mcs$, we obtain 
$\MC(\hat{\beta}) > \mcs$.
\qed
\end{proof}

\noindent
Note that Lemma~\ref{prop:opt-runs-char} does not specify any bound
on the length of $\beta$ and in general, this length can be exponential.
Now we show that an optimal $F$-visiting cycle and an optimal
$S$-visiting cycle can be represented by a counting controller 
constructible in polynomial time. This is the technical core of
our construction which completes the proof of Theorem~\ref{thm:alg}. 

\begin{lemma}
\label{prop:algorithm}
Suppose that $\C$ does not contain any zero-cost cycle, and
let $\T$ be either $S$ or $R$. Then there exist a counting controller 
$\kappa$ and a reload state $r$ computable in polynomial time such 
that $\run{\kappa}{r}=\beta^\omega$ where $\beta$ is
an optimal \T-visiting cycle.
\end{lemma}

%% file: nobuchi.tex
\vspace{-0.3cm}
\subsection{A Proof of Lemma~\ref{prop:algorithm}}
\label{sec:nobuchi}

We start by refining the notion of an optimal \T-visiting cycle and
identifying those cycles that can be represented by counting controllers
of polynomial size. 

A \emph{segment} of a path $\beta$ is a finite subpath $\eta$ of
$\beta$ such that the first and the last state of $\eta$ are reload
states and $\eta$ does not contain any other occurrence of a reload
state. Note that every reload-simple cycle is composed of at most 
$|R|$ segments.
Furthermore, we say that a finite path is \emph{compact}, if it is a
$\Ca$-bounded path of the form $\gamma\cdot \delta^k\cdot \gamma'$,
where $\gamma$ and $\gamma'$ are finite paths satisfying
$\len{\gamma}+\len{\gamma'}\leq 5|S|^3$, $\delta$ is either a cycle of
length at most $|S|$ or a path of length 0 (i.e., a state), and $k\leq
\Ca$. A \emph{compact segment} is a compact path that is also a
segment.

Later we show that there is an optimal \T-visiting cycle $\beta$ such
that every segment of $\beta$ is a compact segment. Intuitively, such
a cycle can be produced by a counting controller of polynomial size
which has at most $|R|$ reset actions. However, this does not
yet imply that such a counting controller can be efficiently constructed,
because there are exponentially many possible compact segments. Hence,
we need to show that we can restrict our attention to some set of
compact segments of polynomial size.

We say that a compact segment $\gamma \cdot\delta^k
\cdot\gamma'$ has a {\em characteristic} $(r,q,t,m,n,b)$,  where
$r,t\in R$, $q\in S$, $m,n\in\Nset$ are such that $0\leq m\leq 5|S|^3
$ and $0\leq n \leq |S| $, and $b\in\{0,1\}$, if the following
holds:
\begin{itemize}
\item $\gamma(0) = r$, $\last(\gamma)=\gamma'(0)=q$,
  $\last(\gamma')=t$, and $\len{\gamma\cdot\gamma'}=m$;
\item $\delta(0) = q$, $\len{\delta}=n$;
\item we either have that $n=0$ and $k=1$, or $n>0$ and then
  $c(\delta)>0$ and $k$ is the maximal number such that $\gamma
  \cdot\delta^k \cdot\gamma$ is a $\Ca$-bounded path;
\item if $b=1$, then $\gamma\cdot\gamma'$ contains a state of $T$;
\item if $\delta$ contains a state of $T$, then $\gamma\cdot\gamma'$
  also contains a state of $T$.
\end{itemize}
Note that for a given consumption system there are at most
polynomially many distinct characteristics of compact segments. Also
note that not all compact segments have a characteristic (because of
the third and the fifth condition in the above definition), 
and conversely, some compact segments may have
multiple characteristics (e.g., if a compact segment has a
characteristic where $b=1$, then it also has one where $b=0$). Finally,
note that for any compact segment $\gamma \cdot\delta^k \cdot\gamma'$
with a characteristic $(r,q,t,m,n,b)$, the path $\gamma\cdot\gamma'$ is a
compact segment with the characteristic $(r,q,t,m,0,b)$.

A characteristic $\chi$ of a compact segment $\gamma\cdot\delta^k\cdot
\gamma'$ imposes certain restrictions on the form of
$\gamma\cdot\gamma'$ and $\delta$. Such a compact segment is
\emph{optimal} for $\chi$ if $\gamma\cdot\gamma'$ and $\delta$ are
paths of minimal cost among those that meet this
restriction.  Formally, a compact segment $\gamma
\cdot\delta^k \cdot\gamma'$ with a characteristic $\chi=(r,q,t,m,n,b)$
is \emph{optimal for} $\chi$ if
\vspace{-0.2cm}
\begin{itemize}
\item $\cost(\gamma\cdot\gamma')$ is minimal among the costs of all
  segments with the characteristic $(r,q,t,m,0,b)$, and
\item $\cost(\delta)$ is minimal among the costs of all cycles of
  length $n$ and positive cost, that are initiated in $q$, and that do
  not contain any reload state with a possible exception of $q$ (if
  $n=0$, we consider this condition to be satisfied trivially).
\end{itemize}

\begin{lemma}
\label{lem:optimal-segments-char}
If there is at least one compact segment with a given characteristic
$\chi$, then there is also an optimal compact segment for
$\chi$. Moreover, all compact segments optimal for a given
characteristic have the same total cost and length.
\end{lemma}

\noindent
Hence, to each of the polynomially many characteristics $\chi$ we can
assign a segment optimal for $\chi$ and thus form a polynomial-sized
candidate set of compact segments.  The following lemma, which is perhaps
the most intricate step in the proof of Lemma~\ref{prop:algorithm}, 
shows that there is an optimal \T-visiting cycle $\beta$ such that
every segment of $\beta$ belongs to the aforementioned candidate set. 
\begin{lemma}
\label{lemma:segment-candidates}
There is an optimal \T-visiting cycle $\beta$ whose every segment is a compact segment optimal for some characteristic. 
\end{lemma}

\noindent
Given a characteristic $\chi$, it is easy to compute a succinct
representation of some compact segment optimal for $\chi$, as the
next lemma shows.

\begin{lemma}
\label{lem:segments-algorithm}
Given a characteristic $\chi$, the problem whether the set of all 
compact segments with the characteristic $\chi$ is non-empty is
decidable in polynomial time. Further, if the set is non-empty, then
a tuple $(\gamma,\gamma',\delta,k)$ such that
$\gamma\cdot\delta^k\cdot\gamma'$ is a compact segment optimal for
$\chi$ is computable in polynomial time.
\end{lemma}
\noindent
For a given characteristic $\chi$, we denote by $\ctuple(\chi)$ the
tuple $(\gamma,\gamma',\delta,k)$ returned for $\chi$ by the algorithm
of Lemma~\ref{lem:segments-algorithm} (if an optimal
compact segment for $\chi$ does not exist, we put
$\ctuple(\chi)=\bot$), and by $\cpath(\chi)$ the corresponding compact
segment $\gamma\cdot\delta^k\cdot\gamma'$ ($\ctuple(\chi)=\bot$, we put
$\cpath(\chi)=\bot$). The next lemma is a simple corollary to
Lemma~\ref{lem:optimal-segments-char} and Lemma~\ref{lemma:segment-candidates}.

\begin{lemma}
\label{col:cycle-with-segments}
There is an optimal \T-visiting cycle $\beta$ such that every
segment of $\beta$ is of the form $\cpath(\chi)$ for some
characteristic $\chi$.
\end{lemma}

\noindent
Now we can easily prove the existence of a polynomial-sized counting
controller representing some optimal \T-visiting cycle
$\beta$. According to Lemma~\ref{col:cycle-with-segments}, there is
a sequence $\chi_0,\chi_1,\dots,\chi_j$ of at most $|R|$
characteristics such that
$\beta=\cpath(\chi_0)\cdot\cpath(\chi_1) \cdots \cpath(\chi_g)$ is an
optimal \T-visiting cycle. To iterate the cycle $\beta$ forever
(starting in $\beta(0)$), a counting controller requires at most
$|R|\cdot n$ basic memory elements, where $n$ is the maximal number of
basic memory elements needed to produce a compact segment $\cpath(\chi_i)$,
for $0\leq i \leq g$. So, consider a compact segment
$\cpath(\chi_i)=\gamma\cdot\delta^k\cdot\gamma'$. Note that $k\leq
\Ca$ since $\cpath(\chi_i)$ has a characteristic and thus
$\cost(\delta)>0$. To produce $\cpath(\chi_i)$, the controller requires
at most $5|S|^3$ basic memory elements to produce the prefix $\gamma$ and
the suffix $\gamma'$ of $\cpath(\chi_i)$, and at most $|S|$ basic memory 
elements to iterate the cycle $\delta$ (whose length is at most $|S|$) 
exactly $k$ times. The latter task also requires counting down 
from $k\leq \Ca$ to $0$. Overall, the counting controller producing 
$\beta^\omega$ needs a polynomial number of basic memory elements, and 
requires at most $|R|$ reset actions parameterized by numbers of 
encoding size at most $\log(\Ca)$. To compute such a counting controller,
it clearly suffices to compute the aforementioned sequence of tuples
$\ctuple(\chi_0),\cdots,\ctuple(\chi_g)$. %

Now we can present the algorithm promised in 
Proposition~\ref{prop:algorithm}. In the following, we use $X$ to denote
the set of all possible characteristics of compact segments in $\C$,
$X_{r,t}$ to denote the set of all characteristics of the form $(r,q,t,m,n,b)$
for some $q,m,n,b$, and $X_{r,t}^1$ to denote the set of all characteristics
of $X_{r,t}$ where the last component is equal to~$1$.
The algorithm first computes the set $R'\subseteq R$ of all admissible
reload states (see Lemma~\ref{lem:algorithm-admissible}). Note that 
$R'$ is non-empty because there exists at least one \T-visiting cycle.
The idea now is to compute, for every $\hat{q}\in R'$, a polynomial-sized labelled graph $G_{\hat{q}}$ such that cycles in this graph correspond to \T-visiting cycles in $\C$ that are initiated in $\hat{q}$ and that can be decomposed into segments of the form $\cpath(\chi)$. An optimal \T-visiting cycle 
is then found via a suitable analysis of the constructed graphs.

Formally, for a given $\hat{q}\in R'$ we construct a labelled graph 
$G_{\hat{q}}=(V,\bbtran{},L,\ell)$, where $L\subset \Nset_0^2$, and where:  
\begin{itemize}
\item $V=W\times\{0,\dots,|S|\}$, where $W=R'\cup\{\ctuple(\chi)\mid
  \chi \in X \}$.
\item For every $0\leq i < |S|$, every pair of states $r,t\in R'$ such
  that $r\neq \hat{q}$, and every characteristic $\chi\in X_{r,t}$
  there is an edge $((r,i),(\cpath(\chi),i))$ labelled by
  $(\cost(\cpath(\chi)),\len{\cpath(\chi)})$ and an edge
  $((\cpath(\chi),i),(t,i+1))$ labelled by $(0,0)$.
\item For every state $t\in R'$ and every characteristic $\chi\in
  X_{\hat{q},t}^1$ there is an edge $((\hat{q},0),(\cpath(\chi),0))$
  labelled by $(\cost(\cpath(\chi)),\len{\cpath(\chi)})$ and an edge
  $((\cpath(\chi),0),(t,1))$ labelled by $(0,0)$.
\item For every $1\leq i \leq |S|$ there is an edge
  $((\hat{q},i),(\hat{q},0))$ labelled by
  $(0,0)$. %
\item There are no other edges.
 \end{itemize}
 The labelling function of $G_{\hat{q}}$ can be computed in polynomial time,
 because given a characteristic $\chi$, we can compute
 $\cpath(\chi)=(\gamma,\gamma',\delta,k)$ using 
 Lemma~\ref{lem:segments-algorithm}. Then, $\len{\cpath(\chi)} =
 \len{\gamma}+\len{\gamma'}+k\cdot\len{\delta}$, and similarly for
 $c(\cpath(\chi))$. Note that every cycle in $G_{\hat{q}}$
 contains the vertex $(\hat{q},0)$. Also note that some of the 
 constructed graphs $G_{\hat{q}}$ may not have a cycle (the out-degree of
 $(\hat{q},0)$ may be equal to $0$), but later we show that at least 
 one of them does.
 
 \vspace{1.1mm} The \emph{ratio} of a cycle $\hat\beta=v_0
 \bbtran{(c_0,d_0)} v_1 \bbtran{(c_1,d_1)} v_2 \cdots
 \bbtran{(c_{h-1},d_{h-1})} v_h$ in $G_{\hat{q}}$ is the value
 $\ratio(\hat\beta)=(c_0 + c_1 + \cdots +c_{h-1})/(d_0 + d_1 + \cdots
 d_{h-1})$.  For every $\hat{q}\in R'$, our algorithm finds a simple
 cycle $\hat{\beta}_{\hat{q}}$ of minimal ratio among all cycles in
 $G_{\hat{q}}$. This is done using a polynomial-time
 algorithm for a well-studied problem of \emph{minimum cycle ratio}
 (see, e.g.,~\cite{DBR:cost-time,DIG:cost-time}). The algorithm then
 picks $\hat{r}\in R'$ such that the ratio of $\hat{\beta}_{\hat{r}}$
 is minimal. Clearly, $\hat{\beta}_{\hat{r}}$ has an even length and
 every second vertex is a $4$-tuple of the form $\ctuple(\chi)$
 for some characteristic $\chi$. Since all cycles in $\hat{r}$ go
 through $(\hat{r},0)$, we may assume that $\hat\beta_{\hat{r}}$ is
 initiated in this vertex. Let
 $\ctuple(\hat\chi_0),\ctuple(\hat\chi_1),\dots,\ctuple(\hat\chi_{g})$
 be the sequence of these \mbox{$4$-tuples}, in the order they appear in
 $\hat{\beta}_{\hat{r}}$. From the construction of $G_{\hat{r}}$ it
 follows that
 $\beta=\cpath(\hat\chi_0)\cdot\cpath(\hat\chi_1)\cdots\cpath(\hat\chi_{g})$
 is a reload-simple cycle initiated in an admissible state $\hat{r}$
 containing a state of $T$ (since $\chi_0$ has the last component
 equal to~$1$), i.e., $\beta$ is a \T-visiting cycle. Moreover, $\MP(\beta)$ is
 clearly equal to the ratio of $\hat{\beta}_{\hat{r}}$. Using the
 computed sequence of tuples
 $\ctuple(\hat\chi_0),\ctuple(\hat\chi_1),\dots,\ctuple(\hat\chi_{g})$,
 the algorithm constructs the desired counting controller $\kappa$ such that
 $\run{\kappa}{\beta(0)}=\beta^\omega$
 (see also the discussion after Lemma~\ref{col:cycle-with-segments}).
 It is easy to check that $\ratio(\hat\beta_{\hat{r}})=\MC(\hat\beta_{\hat{r}})$ 
 is equal to the mean cost of an optimal \T-visiting cycle, i.e., the
 algorithm is correct.

%% file: limit.tex
\vspace{-0.2cm}
\subsection{Proof of Theorem~\ref{thm:limit}}
\label{sec:limit}

For the rest of this section we fix a consumption
system $\C=(S,\tran{},c,R,F)$ and an initial state $s\in S$.
%
%
Intuitively, the controller can approach the limit value by interleaving a large number of iterations of some ``cheap'' cycle with visits to an accepting state.
This motivates the following definitions of \emph{safe} and
\emph{strongly safe} cycles. Intuitively, a cycle is safe if, assuming
unbounded battery capacity, the controller can iterate the cycle for an
arbitrary number of times and interleave these iterations with visits
to an accepting state. A cycle is strongly safe if the same behaviour
is achievable for some finite (though possibly large) capacity.

Formally, we say that two states $q,t \in S$ are
\emph{inter-reachable} if there is a path from $q $ to $t$ and a path
from $t$ to $q$ (i.e., $q,t$ are in the same strongly connected
component of $\C$). We say that a cycle $\beta$ of length at most
$|S|$ where $\beta(0)$ is reachable from $s$ is \emph{safe}, if one of the
following conditions holds:
\begin{itemize}
\item  $\cost(\beta)=0$ and $\beta$ contains an accepting state, 
\item $\beta(0)$ is inter-reachable with a reload state and an accepting state, %
\end{itemize}
A cycle $\beta$ reachable from $s$ with $\len{\beta}\leq|S|$ is \emph{strongly safe}, if
one of the following holds:
\begin{itemize}
\item  $\cost(\beta)=0$ and $\beta$ contains an accepting state,
\item $\cost(\beta)=0$ and $\beta(0)$ is inter-reachable with a reload 
state and an accepting state,
\item  $\beta$ contains a reload state and $\beta(0)$ is inter-reachable with an accepting state.
\end{itemize}

\noindent
The following lemma characterizes the limit value of $s$.

\begin{lemma}
\label{lem:lim-val-char}
$\Val_\C(s)$ is finite iff there is a safe cycle, in which 
case  $\Val_\C(s)= \min\{\MC(\beta)\mid \beta \text{ is a safe cycle}  \}$.
Further, there is a finite $\Ca\in\Nset_0$ such that
$\Val^\Ca_\C(s)=\Val_\C(s)$ iff either $\Val_\C(s)=\infty$,
or there is a strongly safe cycle $\hat\beta$ such that
$\MC(\hat\beta)=\Val_\C(s)$. %
In such a case $\Val^\Ca_\C(s)=\Val_\C(s)$ for every $\Ca\geq
3\cdot|S|\cdot \cmax$, where $\cmax$ is the maximal cost of a
transition in $\C$.
\end{lemma}

\noindent
So, in order to compute the limit value and to decide whether it can be
achieved with some finite capacity, we need to compute a safe and a
strongly safe cycle of minimal mean cost.

\begin{lemma}
\label{lem:limit-algorithms}
The existence of a safe (or strongly safe) cycle is decidable in polynomial
time. Further, if a safe (or strongly safe) cycle exists, then there
is a safe (or strongly safe) cycle $\beta$ computable in polynomial time
such that $\MC(\beta) \leq \MC(\beta')$ for every safe (or strongly safe)
cycle $\beta'$.
\end{lemma}

\noindent
Now we can prove the computation-related statements of
Theorem~\ref{thm:limit}.

To compute the limit value of $s$, we
use the algorithm of Lemma~\ref{lem:limit-algorithms} to compute a
safe cycle $\beta$ of minimal mean cost. If no such cycle exists,
we have $\Val_\C(s)=\infty$, otherwise $\Val_\C(s)=\MC(\beta)$. 
To decide whether $\Val_\C(s)$  can be achieved
with some finite capacity, we again use the algorithm of
Lemma~\ref{lem:limit-algorithms} to compute a strongly safe cycle
$\hat\beta$ of minimal mean cost. If such a cycle exists and
$\MC(\hat\beta)=\MC(\beta)$, then $\Val_\C(s)$ can be achieved with
some finite capacity, otherwise not. The correctness of this approach
follows from Lemma~\ref{lem:lim-val-char}.

It remains to bound the rate of convergence to the limit value in case
when no finite capacity suffices to realize it. This is achieved in the
following lemma.

\begin{lemma}
\label{lem-limit-rate}
  Let $\cmax$ be the maximal cost of a transition in $\C$. For every
  $\Ca> 4\cdot|S|\cdot\cmax$ we have that
  $$\Val_\C^\Ca(s)-\Val_\C(s)\leq
  \frac{3\cdot|S|\cdot\cmax}{\Ca-4\cdot|S|\cdot \cmax}.$$
\end{lemma}
%
%
%
%
%
%
%
%
%
%
%
%
%
%
%
%
%
%
%
%
%
%
%
%
%
%
%
%

%
%
%
%

%% file: concl.tex
\section{Future work}

We have shown that an optimal controller for a given consumption system 
always exists and can be efficiently computed. We have also exactly 
classified the structural complexity of optimal controllers and analyzed the
limit values achievable by larger and larger battery capacity.

The concept of $\Ca$-bounded mean-payoff is natural and generic, and 
we believe it deserves a deeper study. Since mean-payoff has been 
widely studied (and applied) in the context of Markov decision processes,
a natural question is whether our results can be extended to
MDPs. Some of our methods are surely applicable, but the 
question appears challenging.

%
%
%
%

%% file: app-example.tex
\section{Non-trivial Behaviour of Optimal Controllers}
\label{app-example}

Consider the consumption system $\C$ on Figure~\ref{fig-nontriv}, where $R=\{s\}$ and $F=S$.

\begin{figure}[h]
\begin{center}
\begin{tikzpicture}[x=2cm,y=2cm,font=\small]
\node (t) at (-1,-1) [ran] {$q_1$};
\node (v) at (-3,0) [ran] {$q_3$};
\node (w) at (-1,1) [ran] {$q_5$};
\node (s) at (0,1) [ran] {$s$};
\node (u) at (0,0) [ran, very thick] {$u$};
\node (r) at (1,0) [ran] {$r$};
\node (q2) at (-2,-1) [ran] {$q_2$};
\node (q4) at (-2,1) [ran] {$q_4$};
\node (nt) at (0,-1) [ran] {$t$};
\draw [tran] (u.300) -- node[auto] {$349$} (nt.60);
\draw [tran] (nt.120) -- node[auto] {$0$} (u.240);
\draw  (nt) edge[tran,loop, in=300, out=240, looseness=5] node[below] {$1$} (nt);
\draw [tran] (q2.135) -- node[auto] {$0$} (v.315);
\draw [tran] (v.45) -- node[auto] {$0$} (q4.225);
\draw [tran] (q4.0) -- node[auto] {$0$} (w.180);
\draw [tran] (t.180) -- node[auto] {$0$} (q2.0);
\draw [tran] (w.335) -- node[auto,swap] {$0$} (u.135);
\draw [tran] (s.240) -- node[left, near start] {$50$} (u.120);
\draw [tran] (u.60) -- node[right, near end] {$50$} (s.300);
\draw [tran] (u.235) -- node[auto,swap] {$60$} (t.45);
\draw [tran] (r.135) -- node[above] {$0$} (u.45);
\draw [tran] (u.315) -- node[below] {$22$} (r.225);
\end{tikzpicture}
\end{center}
\caption{A consumption system with a non-trivial behaviour of an optimal controller.}
\label{fig-nontriv}
\end{figure}
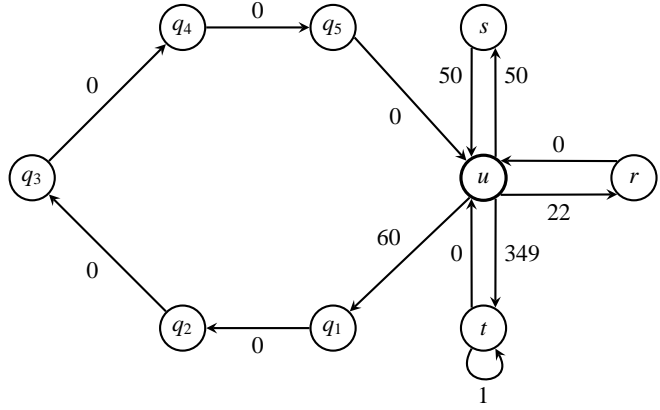

Denote by $\delta_1$ the cycle $uq_1q_2q_3q_4q_5u$, and by $\delta_2$ the cycle $uru$. Let $s$ be the initial state and let $\Ca=450$. 
Clearly, a controller that visits $t$ cannot be optimal, because of the enormous mean cost of the cycle $utu$. Moreover it does not make sense to iterate the cycle $sus$, since its mean cost is much larger than the mean cost of both $\delta_1$ and $\delta_2$. Thus, an optimal controller goes from $s$ to $u$, then iterates $\delta_1$ for $A\in \{0,1,\dots,5\}$ times, then possibly iterates $\delta_2$ for $B \in \{0,1,\dots,15\}$ times, and then returns to $s$. An easy computation shows that the optimal behaviour is achieved for $A=5$ and $B=2$, and the resulting mean cost is equal to $37/3$. This shows that the optimal controller generally has to iterate more than one simple cycle between two visits of a reload state (the controller from example on Figure~1 iterated only 1 simple cycle before returning to the reload state). Also, note that the cycle $tt$, which has the minimal mean cost among all simple cycles in $\C$, is not traversed by the optimal controller at all.

\begin{figure}[t]
\begin{center}
\begin{tikzpicture}[x=2.2cm,y=0cm,font=\small]
\node (t) at (0,0) [ran] {$t$};
\node (s) at (1,0) [ran] {$s$};
\node (u) at (2,0) [ran, very thick] {$u$};
\draw [tran] (t.45) -- node[above] {$1$} (s.135);
\draw [tran] (s) -- node[above] {$0$} (u);
\draw [tran] (s.225) -- node[below] {$0$} (t.315);
\draw (u) edge[tran, loop, looseness=5, out = 45, in=-45]   node[right] {$0$} (u);
\end{tikzpicture}
\end{center}
\caption{Limit value is not equal to the B\"uchi mean-payoff value. Here $R=\{u\}$ and $F=\{t\}$.}
\label{fig-limit-nobuchi}
\end{figure}
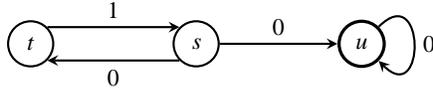

Now let us again consider the example on Figure~\ref{fig-exp-mem}. Note that for any capacity $\Ca\geq 10$ we have $\Val_{\C}^{\Ca}(s)=\frac{\Ca}{2+2\cdot(\Ca-10)}=\frac{1}{2}\cdot\frac{\Ca}{\Ca-9}$. It is then clear that $\Val_\C(s)=1/2$ and that this limit value cannot be achieved for any finite capacity. 

Finally, consider the consumption system $\C$ on Figure \ref{fig-limit-nobuchi}. For every capacity $\Ca$ we have $\Val_{\C}^{\Ca}(s)=\infty$, since every $\Ca$-bounded path $\C$ must have an infinite suffix $u^\omega$ and thus it cannot be accepting. Thus, the limit value $\Val_{\C}(s)$ is also infinite. However, if we treat the system as a one-player mean-payoff B\"uchi game (see, e.g., \cite{CHJ:MP-parity-games}), then the optimal value in $s$ is $1/2$. 

%% file: app-cap-reach-alg.tex
\section{Proofs for Auxiliary Algorithms in Section~\ref{sec:main-thm-proof}}
\label{app-main}

In this section we provide the proofs for auxiliary algorithms from Section~\ref{sec:main-thm-proof}, i.e., the proof of Lemmas~\ref{lem:infinite-value},\ref{lem-cap-bounded-path},\ref{lem:algorithm-admissible} and~\ref{lem-cycle-compute}. The problems solved by these algorithms ($\Ca$-reachability, existence of an acceptance witness, etc.) are variants of standard graph-theoretic problems. Our choice of algorithms is motivated by our intention to achieve as straightforward correctness proofs and proofs of polynomial running time as possible. It is not hard to see that the complexity of these algorithms (which does not dominate the overall complexity of the algorithm from Theorem~\ref{thm:alg}) can be improved.

Note that the lemmas are not proved in the order in which they appear in the main text. We first prove Lemma~\ref{lem-cap-bounded-path}, then Lemma~\ref{lem:algorithm-admissible}, then Lemma~\ref{lem:infinite-value} and finally Lemma~\ref{lem-cycle-compute}. This is because some algorithms use as a sub-procedure an algorithm which, in the main text, appears, later than the algorithm in which this sub-procedure is used. We chose to mention the algorithms in the main text in a different order to make the flow of ideas in the main text more natural.

\subsection{A Proof of Lemma~\ref{lem-cap-bounded-path}}

We denote by $\ec(\alpha)$ the end cost of a finite path $\alpha$. First we prove the following lemma.

\begin{lemma}
\label{lem:cap-bounded-length}
If a state $t$ is $\Ca$-reachable from $s$, then there is a $\Ca$-bounded path $\alpha$ of length at most $|S|^2$ such that $\ec(\alpha)=\min\{\ec(\alpha')\mid \alpha' \text{ is a $\Ca$-bounded path from $s$ to $t$}\}$.
\end{lemma}
\begin{proof}
Let $t$ be $\Ca$-reachable from $s$. 
Since the end costs are natural numbers, the value $\minec=\min\{\ec(\alpha')\mid \alpha' \text{ is a $\Ca$-bounded path from $s$ to $t$}\}$ exists. Let $\alpha$ be a $\Ca$-bounded path of minimal length among the $\Ca$-bounded paths from $s$ to $t$ that have $\ec(\alpha)=\minec$. Assume, for the sake of contradiction, that $\len{\alpha}>|S|^2$. %
Two cases may arise: either at most $|R|$ reload states appear on $\alpha$ and then $\alpha=\gamma\cdot\delta\cdot\gamma',$ where either $\delta$ is a cycle not containing a reload state, or $\delta(0)=\last(\delta)\in R$.  In both cases clearly, $\gamma\cdot\gamma'$ is a $\Ca$-bounded path from $s$ to $t$ with $\ec(\gamma\cdot\gamma')=\ec(\alpha)$, a contradiction with the choice of $\alpha$. The second case is that there are more than $|R|$ occurrences of a reload state on $\alpha$. Then $\alpha=\gamma\cdot\delta\cdot\gamma'$, where $\delta(0)=\last(\delta)\in R$. As above, we get a contradiction with the choice of $\alpha$.
\qed
\end{proof}

\noindent
We now prove a slightly more general variant of Lemma~\ref{lem-cap-bounded-path}.

\begin{lemma}
\label{lem:cap-bounded-extended}
 There is a polynomial-time algorithm, which for a given state $s$, given capacity $\Ca$ and every state $t\in S$ decides, whether $t$ is $\Ca$-reachable from $s$. Moreover, for every $t$ that is $\Ca$-reachable from $s$ the algorithm computes a $\Ca$-bounded path $\alpha$ of length at most $|S|^2$ such that $\ec(\alpha)=\min\{\ec(\alpha')\mid \alpha' \text{ is a $\Ca$-bounded path from $s$ to $t$}\}$.
\end{lemma}
\begin{proof}
Denote $\minec=\min\{\ec(\alpha')\mid \alpha' \text{ is a $\Ca$-bounded path from $s$ to $t$}\}$.
Moreover, for any $t\in S$ and any $i\in \Nset_0$ we denote $$\minec^i(t)=\min\{\ec(\alpha)\mid \alpha\text{ is a $\Ca$-bounded path from $s$ to $t$ of length $i$}\}.$$ From Lemma~\ref{lem:cap-bounded-length} it follows that $\minec(t)=\min_{0\leq i \leq |S|^2}\minec^i(t)$.

Now consider an operation $\plusca$ on $\Nset_0 \cup\{\infty\}$ such that for any $a\plusca b = a+b$ if $a+b\leq \Ca$ and $\infty$ otherwise (we use a standard convention that $\infty+x=x+\infty=\infty$ for any $x\in \Nset_0 \cup\{\infty\}$). Clearly $\minec^0(t)$ is equal to 0 if $t=s$ and equal to $\infty$ otherwise. For $i>0$ a straightforward induction reveals that 
\[
\minec^i(t) = \begin{dcases}
\min_{q\tran{a} t}\,(\minec^{i-1}(q) \plusca a ) & \text{if } t \not\in R \\
0 & \text{if } t \in R \text{ and } \min_{q\tran{a} t}\, (\minec^{i-1}(q) \plusca a ) <\infty \\
\infty & \text{otherwise }.
\end{dcases}
\]
Using these equations we can iteratively compute $\minec^0(q),\minec^1(q),\dots,\minec^{|S|^2}(q)$ for all states $q$ in polynomial time. Now by Lemma~\ref{lem:cap-bounded-length} we have for any  state $t$ that $\minec(t)=\min_{0\leq i \leq |S|^2}\minec^i(t)$, and clearly $t$ is $\Ca$-reachable from $s$ iff $\minec(t)<\infty$. Moreover, let $j_t\leq |S|^2$ be such that $\minec^{j_t}(t)=\minec(t)<\infty$. Then from the knowledge of $\minec^0,\minec^1,\dots,\minec^j$ we can construct a finite $\Ca$-bounded path  $\alpha=q_0 q_1 \dots q_{j_t}$ with from $s$ to $t$ by putting $q_{j_t}=t$ and for every $j<j_t$ defining $q_j$ the state that caused the $\minec^{j+1}(q_{j+1})$ to be set to its final value. I.e., if $q_{j+1}\not\in R$, then $q_j$ is such that $q_j \tran{a} q_{j+1}$ and $\minec^j(q_j)\plusca a = \minec^{j+1}(q_{j+1})$, otherwise $q_j$ is such that $q_j \tran{a} q_{j+1}$ and $\minec^j(q_j)\plusca a <\infty$. The correctness of this approach is immediate.
\qed
\end{proof}

\subsection{A Proof of Lemma~\ref{lem:algorithm-admissible}}

Before we prove Lemma~\ref{lem:algorithm-admissible}, we prove the following simple lemma.

\begin{lemma}
\label{lem:cycle-of-zero-cost}
For a given state $s$ it is decidable in polynomial time whether there is a cycle $\delta$ of zero cost containing $s$, and if such a $\delta$ exists, we can compute in polynomial time a simple cycle with this property. Moreover, for a given set of states $T\subseteq S$ it is decidable in polynomial time, whether there is a cycle of zero cost containing $s$ and a state from $T$, and if such a cycle exists, we can compute such a cycle of length at most $2|S|$ in polynomial time.
\end{lemma}
\begin{proof}
Clearly if there is a cycle of zero cost on $s$, there is also a simple cycle of zero cost on $s$. Such a cycle can be find using a simple modification of e.g., the breadth-first search algorithm -- for every $t$ such that $s\tran{0}t$ we try to find a path from $t$ to s while ignoring the transitions of positive cost. Similarly, if there is a cycle of zero cost on $s$ that contains a state from $T$, then the shortest such cycle has length at most $2|S|$ (we need to get from $s$ to a state $t\in T$ and back via transitions of zero cost, each of these two parts requiring at most $|S|$ transitions). Again, such a cycle can be find using a suitable search algorithm: first, we compute a set $T'\subseteq T$ of states that are reachable via transitions of zero cost from a state $q$ s.t. $s\tran{0}q$ , and then, for every $t\in T'$ we try to find (possibly empty) path from $t $ to $s$. Clearly, both tasks can be implemented using a simple graph search algorithm. 
\qed
\end{proof}

\noindent
The following lemma will be also useful.
\begin{lemma}
\label{lem:interreaching-cycle}
Let $r\in R$ be a reload state and $T\subseteq S$ a set of states. It is decidable in polynomial time whether there is a $\Ca$-bounded cycle $\delta$ that is initiated in $r$ and that contains a state from $T$. If the answer is yes, one can compute (in polynomial time) such a cycle $\delta$ of length at most $3|S|^2$.
\end{lemma}
\begin{proof}
Recall that we denote by $\ec(\alpha)$ the end cost of a finite path $\alpha$. Denote by $\minec_r(t)$ the value $\min\{\ec(\alpha)\mid \alpha \text{ is a $\Ca$-bounded path from $r$ to $t$ } \}$. 
We claim that a $\Ca$-bounded $\delta$ initiated in $r$ and containing a state from $T$ exists if and only if there is a $t\in T$ and $q\in R$ (possibly $q=r$) such that:
\begin{itemize}
\item There is a $\Ca$-bounded path $\gamma_1$ from $r$ to $t$, and
\item there is a $(\Ca-\minec_r(t))$-bounded path $\gamma_2$ from $t$ to $q$, and
\item there is a $\Ca$-bounded path $\gamma_3$ from $q$ to $r$. 
\end{itemize}
The ``if'' direction can be proved as follows: if there are the aforementioned paths, then there is in particular a $\Ca$-bounded path $\gamma_1$ from $r$ to $t$ such that $\ec(\gamma_1)=\minec_r(t)$ (see Lemma~\ref{lem:cap-bounded-extended}). Then $\delta=\gamma_1\cdot\gamma_2\cdot\gamma_3$ is the required $\Ca$-bounded cycle containing a state from $T$.

Now consider the ``only if'' direction in the equivalence. Let $i\in \Nset_0$ be the smallest number such that $\delta(i)\in T$ and $j>i$ be the smallest number such that $\delta(j)\in R$ (such $i,j$ must exist, since $\delta$ is a cycle). Put $\gamma_1 = \delta_{\leq i}$, $\gamma_2 = \delta(i)\dots\delta(j)$ and $\gamma_3=\delta(j)\dots \delta(\len{\delta})$. Clearly $\gamma_3$ and $\gamma_1$ are $\Ca$-bounded and $\cost(\gamma_2)\leq \Ca-\ec(\gamma_1) \leq \Ca - \minec_r(\delta(i))$, which proves the ``only if'' direction.

So in order to decide whether a desired cycle $\delta$ exists (and compute it if it does) it suffices to make three calls of the algorithm from Lemma~\ref{lem:cap-bounded-extended} for every pair of states $t \in T$, $q\in R$.  For every such pair we first use that algorithm to compute a $\Ca$ bounded path $\gamma_1$ from $r$ to $t$ of minimal end cost (and thus also compute $\minec_r(t)$). Then we use the algorithm to find  a $(\Ca-\minec_r(t))$-bounded path $\gamma_2$ from $t$ to $q$, and a $\Ca$-bounded path $\gamma_3$ from $q$ to $r$. If we find all these paths, we return $\gamma_1\cdot\gamma_2\cdot\gamma_3$ as the desired cycle $\delta$ (from Lemma~\ref{lem:cap-bounded-extended} it follows that $\len{\delta}\leq 3|S|^2$). If some of these paths does not exist, we move on to the next pair. If the algorithm fails for all pairs, the desired cycle $\delta$ does not exist. The correctness of the algorithm and its polynomial complexity are immediate.
\qed
\end{proof}

\begin{reflemma}{lem:algorithm-admissible}
The problem whether a given $q \in S$ is admissible is decidable
in polynomial time. Further, if $q$ is admissible, then 
there are finite paths $\alpha,\gamma$ computable in polynomial time
such that $\alpha \cdot \gamma^\omega$ is a $\Ca$-bounded run initiated 
in $s$ and $\gamma$ is an admissibility witness for $q$ of length
at most $6|S|^2$.
\end{reflemma}
\begin{proof}
First we prove that a state $q\in S$ is admissible if and only if $q$ is $\Ca$-reachable from the initial state $s$ and one of the following conditions holds.
\begin{enumerate}
\item There is a cycle $\delta$ of zero cost that is initiated in $q$ and that contains an accepting state. In this case, $\delta$ is an admissibility witness for $q$.
\item There is a reload state $r\in R$ and $\Ca$-bounded cycles $\theta$, $\delta$, both initiated in $r$, such that $\delta$ contains an accepting state and $\theta=\theta_1\cdot\theta_2$, where $\theta_1(0)=\last(\theta_2)=r$ and $\last(\theta_1)=\theta_2(0)=q$. In such a case $\theta_2\cdot\delta\cdot\theta_1$ is an admissibility witness for $q$.
\end{enumerate}
The ``if'' direction is immediate, so let us consider the ``only if'' direction. Suppose that $q$ is admissible, then by definition there is a $\Ca$-bounded run initiated in $s$ of the form $\alpha\cdot\beta^\omega$, where $\beta$ is a cycle initiated in $q$ which contains an accepting state. In particular $q$ is $\Ca$ reachable from $s$. Now if $\cost(\beta)=0$, then the case 1. above holds. So suppose that $\cost(\beta)\neq 0$. Then $\beta$ must contain not only an accepting state, but also a reload state. Let $i\in \Nset_0$ be such that $\beta(i)\in R$, respectively. Then the paths $\theta_1 = \beta(i)\beta(i+1)\dots\beta(\len{\beta})$, $\theta_2 = \beta_{\leq i}$, and $\delta=\beta(i)\beta(i+1)\dots\beta(\len{\beta})\cdot\beta_{\leq_i}$ have the properties stated in case 2.

So to test whether $q$ is admissible, we have to test whether $q$ is $\Ca$-reachable from $s$ and whether 1. or 2. holds. To test the $\Ca$-reachability we use the polynomial algorithm of Lemma~\ref{lem:cap-bounded-extended}, which also finds the required $\Ca$-bounded path from $s$ to $q$. To test whether 1. holds, we use the polynomial algorithm from Lemma~\ref{lem:cycle-of-zero-cost}. If this algorithm finds a cycle of zero cost initiated in $q$ and containing an accepting state (by Lemma~\ref{lem:cycle-of-zero-cost}, the cycle returned by the algorithm has length at most $2|S|$), we can immediately output it as an admissibility witness for $q$. To test whether 2. holds, we use $|R|$ times the polynomial algorithm of Lemma~\ref{lem:interreaching-cycle}: For every state $r\in R$ we test whether there  are $\Ca$-bounded cycles $\theta,\delta$ initiated in $r$ such that $\theta$ contains $q$ and $\delta$ contains a state from $F$. If the algorithm finds such cycles for some $r\in R$ (by Lemma~\ref{lem:interreaching-cycle} each of them will have length at most $3|S|^2$), we can use them to easily construct an admissibility witness for $q$ of length at most $6|S|^2$ as indicated in 2. The correctness of the algorithm and its polynomial complexity are immediate.
\qed
\end{proof}

\subsection{A Proof of Lemma~\ref{lem:infinite-value}}

\begin{reflemma}{lem:infinite-value}
Let $t \in S$. The problem whether $\Val_\C^\Ca(t)=\infty$ is
  decidable in polynomial time.
\end{reflemma}
\begin{proof}
Let $t\in S$ be an arbitrary. In the following we treat $t$ as the initial state of the system. In particular, the notion of admissibility is adapted to this choice of initial state: a state $q$ is admissible if there is a $\Ca$-bounded path of the form $\alpha\cdot\gamma^\omega$ with $\alpha(0)=t$, $\gamma(0)=\last(\alpha)=q$, and $\gamma$ containing an accepting state.

First note that $\Val_{\C}^{\Ca}(t)<\infty$ if and only if there is at least one admissible state $q$ 
The ``if'' direction is immediate, so let us consider the ``only if'' direction. If $\Val_{\C}^\Ca(t)<\infty$, then there is a $\Ca$-bounded accepting run $\rho$ initiated in $t$. We consider two cases. Either there are only finitely many transitions of a positive cost on $\rho$. Then there is a simple cycle $\delta$ of zero cost containing an accepting state (since some accepting state has infinitely many occurrences on $\rho$) initiated in some state $\delta(0)$ that is $\Ca$-reachable from $t$. Clearly $\alpha\cdot\delta^\omega$, where $\alpha$ is a $\Ca$-bounded path from $t$ to $q$, is a $\Ca$-bounded accepting run, so $\delta(0)$ is admissible. The second case is that there are infinitely many transitions of positive cost on $\rho$, in which case $\rho$ contains infinitely many occurrences of both a reload state and of an accepting state. Let $r$ be a reload state appearing infinitely often on $\rho$. Then there is a $\Ca$-bounded cycle $\delta$ initiated in $r$ and containing an accepting state. For any $\Ca$-bounded path $\alpha$ from $t$ to $r$ (at least one exists due to the existence of $\rho$) the run $\alpha\cdot\delta^\omega$ is a $\Ca$-bounded accepting run, showing that $r$ is admissible.

So to decide whether $\Val_\C^\Ca(t)=\infty$, it suffices to use the polynomial-time algorithm of Lemma~\ref{lem:cap-bounded-extended} to compute the set of states that are $\Ca$-reachable from $t$, and for every such state decide, whether it is admissible, using the polynomial-time algorithm from Lemma~\ref{lem:algorithm-admissible}. The correctness and the polynomial running time of this procedure are immediate.
\qed
\end{proof}

\subsection{A Proof of Lemma~\ref{lem-cycle-compute}}

\begin{reflemma}{lem-cycle-compute}
The existence of a zero-cost cycle is decidable in polynomial  
  time, and an example of a zero-cost cycle $\beta$ (if it exists)
  is computable in polynomial time. The same holds for zero-cost cycles
  containing an accepting state.
\end{reflemma}
\begin{proof}
Note that a zero-cost cycle, or zero-cost cycle containing an accepting state, is simply a simple cycle of zero cost initiated in an admissible state, or a simple cycle of zero cost containing an accepting state that is initiated in an admissible state, respectively. 

So to decide whether there is a zero-cost cycle, it suffices to compute the set $A\subseteq S$ of all admissible states using the polynomial-time algorithm of Lemma~\ref{lem:algorithm-admissible}, and then for every $q\in A$ try to find a simple cycle of zero cost initiated in $q$ using the polynomial-time algorithm from Lemma~\ref{lem:cycle-of-zero-cost}. If we find such a cycle we can output it as $\beta$, otherwise we conclude that there is no such cycle.

To decide whether there is a zero-cost cycle containing an accepting state, it suffices, for every $q\in A$, to use the polynomial-time algorithm of Lemma~\ref{lem:cycle-of-zero-cost} to find a cycle of zero cost initiated in $q$ containing an accepting state. Note that if the algorithm finds such a cycle $\delta$ for some $q\in A$, this cycle does not have to be simple. However, if it is not simple, then there is a simple cycle $\theta$ of zero cost such that $\delta=\gamma\cdot\theta\cdot\gamma'$ for some finite paths $\gamma,\gamma'$ and $\theta$ contains an accepting state. Moreover, $\delta(0)$ is an admissible state, since $\delta\cdot\gamma'\cdot\gamma$ is an admissibility witness for $\delta(0)$. So $\theta$ is a zero-cost cycle and we can return it as the desired cycle $\beta$ (note that $\theta$ can be easily computed once $\delta$ is computed). If the algorithm of  Lemma~\ref{lem:cycle-of-zero-cost} fails to find a cycle of zero cost with an accepting state for every $q\in A$, we conclude that there is no such cycle.
\qed
\end{proof}

%% file: app-algorithms.tex
\section{Auxiliary Results}
\label{app:algorithms}

This section contains some auxiliary algorithms that are not mentioned in the main text and that will be useful in later proofs.

\begin{lemma}
\label{lem:length-algorithm}
Let $\C=(S,\tran{},c,R,F)$ be a consumption system. There is an algorithm  $\minpath(s_1,s_2,m,\avoid)$ which for a given pair of states $s_1,s_2\in S$, given number $m\in \Nset_0$ and a given set $\avoid\subseteq S$ %
 decides, whether there is a path $\alpha$ satisfying the following conditions:
 \begin{itemize}
 \item $\alpha(0)=s_1$,  $\last{(\alpha)}=s_2$, $\len{\alpha}=m$, and
 \item for all $0<i<\len{\alpha}$ it holds $\alpha(i)\not\in\avoid$.
 \end{itemize}
 If there is such a path, the algorithm computes a path $\alpha$ of minimal cost among all paths satisfying the above conditions. The algorithm runs in time polynomial in $\size{\C}$ and $m$. 
\end{lemma}
\begin{proof}
The algorithm constructs a labelled graph $G=(V,\bbtran{},L,\ell)$, where $L\subset \Nset_0$, is defined as follows:
\begin{itemize}
\item $V= (S \setminus \avoid)\times\{1,2,\dots,m-1\} \cup \{(s_1,0),(s_2,m) \}$.
\item There is an edge $((s,i),(s',j))$ in $G$ if and only if $j=i+1$ and $s\tran{a}s'$ is a transition in $\C$. In such a case $\ell((s,i),(s',j)) = a$. 
\end{itemize}
Then the algorithm finds a path of minimal cost from $(s_1,0)$ to $(s_2,m)$ (or decides that such a path does not exist) using, e.g., the algorithm for computing shortest paths in directed acyclic graphs. The procedure then returns the corresponding path in $\C$ (it suffices to discard the second components from the computed path in $G$). The correctness of the procedure and its complexity analysis are straightforward.
\end{proof}

\begin{lemma}
\label{lem:length-reach-algorithm}
Let $\C=(S,\tran{},c,R,F)$ be a consumption system. There is an algorithm  $\minpathr(s_1,s_2,m,\avoid,\reach)$ which for a given pair of states $s_1,s_2\in S$, given number $m\in \Nset_0$ and given sets $\reach,\avoid\subseteq S$ decides, whether there is a path $\alpha$ satisfying the following conditions:
\begin{itemize}
\item $\alpha(0)=s_1$,  $\last{(\alpha)}=s_2$, $\len{\alpha}=m$, and
\item for all $0<i<\len{\alpha}$ it holds $\alpha(i)\not\in\avoid$, and
\item there is $0\leq j\leq\len{\alpha}$ such that $\alpha(j)\in\reach$.
\end{itemize}
If there is such a path, the algorithm computes a path $\alpha$ of minimal cost among all paths satisfying the above conditions. The algorithm runs in time polynomial in $\size{\C}$ and $m$. 
\end{lemma}
\begin{proof}
If $s_1\in \reach$ or $s_2\in\reach$, then we just call the algorithm $\minpath(s_1,s_2,m,\avoid)$ from Lemma~\ref{lem:length-algorithm}.

Otherwise for every $q\in \reach$ and every $0< i < m$ the algorithm constructs a labelled graph $G_{q,i}=(V,\bbtran{},L,\ell)$, where $L\subset \Nset_0$, is defined as follows: 
\begin{itemize}
\item $V= (S \setminus \avoid)\times\{1,2,\dots i-1,i+1,\dots,m-1\} \cup \{(s_1,0),(s_2,m),(q,i) \}$.
\item There is an edge $((s,i),(s',j))$ in $G$ if and only if $j=i+1$ and $s\tran{a}s'$ is a transition in $\C$. In such a case $\ell((s,i),(s',j)) = a$. 
\end{itemize}
Then, for every $q\in \reach$ and every $0<i<m$ the algorithm finds a path $\alpha_{q,i}$ of minimal cost from $(s_1,0)$ to $(s_2,m)$ in $G_{q,i}$ (or decides that such a path does not exist) using, again the algorithm for computing shortest paths in directed acyclic graphs. 
If $\alpha_{q,i}$ exists for at least for one pair $q,i$, the algorithm returns $\alpha_{q,i}$ of minimal cost (the minimum is taken among all $q\in \reach$, $0<i<m$), otherwise the path satisfying the required conditions does not exist. Again, the correctness of the algorithm and its complexity are straightforward, since every path $\alpha$ satisfying the required conditions induces, in a natural way, a corresponding path $\alpha'$ of the same cost in $G_{q,j}$, where $q,j$ are such that $q=\alpha(j)\in \reach$. Conversely, every $\alpha'$ in some $G_{q,i}$ induces a path $\alpha$ of the same cost in $\C$ that satisfies the required conditions.
\qed
\end{proof}

\section{Proofs of Section~\ref{sec:nobuchi}}

\subsection{A Proof of Lemma~\ref{lem:optimal-segments-char}}

\begin{reflemma}{lem:optimal-segments-char}
If there is at least one compact segment with a given characteristic
$\chi$, then there is also an optimal compact segment for
$\chi$. Moreover, all compact segments optimal for a given
characteristic have the same total cost and length.
\end{reflemma}
\begin{proof}
Fix a characteristic $\chi=(r,q,t,m,n,b)$. 
If there is at least one compact segment
$\eta=\xi\cdot\theta^j\cdot\xi'$ having characteristic $\chi$, there
is also at least one segment of a characteristic $(r,q,t,m,0,b)$
(namely $\xi\cdot\xi'$); and (provided that $n>0$) at least one cycle
of length $n$ initiated in $q$ which is either a segment or does not
contain any reload state (namely $\theta$).  So there also is a
segment $\gamma\cdot\gamma'$ and (provided that $n>0$) a cycle
$\delta$ satisfying the above conditions whose costs are minimal among
all segments and cycles that satisfy these conditions,
respectively. Let $k$ be either $1$ (if $n=0$) or the maximal number
such that $\gamma\cdot\delta^k\cdot\gamma'$ is a $\Ca$-bounded path
(if $n>0$ -- then such a $k$ must exist, because
$\cost(\gamma\cdot\delta\cdot\gamma')\leq \cost(\eta)$). Clearly
$\gamma\cdot\delta^k\cdot\gamma'$ is a compact segment optimal for
$\chi$.

For the second part, let $\eta=\xi\cdot\theta^j\cdot\xi'$ and
$\eta'=\gamma\cdot\delta^k\cdot\gamma'$ be two segments optimal for
the same characteristic $\chi=(r,q,t,m,n,b)$. If $n=0$, then clearly
$\len{\eta}=\len{\eta'}=m$ and from the optimality of both segments we
get the equality of their costs. Otherwise, by definition of optimal
segments we have $\cost{(\gamma\cdot\gamma')}=\cost{(\xi\cdot\xi')}$
and $\cost{(\theta)}=\cost{(\delta)}$. To prove the lemma it suffices
to show that $j=k$. Suppose that, e.g., $j<k$, the other case is
symmetrical. Then $\xi\cdot\theta^{j+1}\cdot\xi'$ is a $\Ca$-bounded
path (since its cost is at most the cost of $\eta'$), a contradiction
with the fact that $\eta$ has a
characteristic. %
\qed
\end{proof}

\subsection{A Proof of Lemma~\ref{lemma:segment-candidates}}
\label{app:segment-candidates}

\begin{reflemma}{lemma:segment-candidates}
  There is an optimal \T-visiting cycle $\beta$ whose every segment
  is a compact segment optimal for some characteristic.
\end{reflemma}
\begin{proof}
  We say that a segment is \emph{bad} if it is not a compact segment
  optimal for some characteristic $\chi$. Given an optimal
  \T-visiting cycle $\beta$ containing $g>0$ bad segments, we show
  how to construct an optimal \T-visiting cycle $\beta'$ containing
  $g-1$ bad segments. Combined with the existence of at least one
  optimal \T-visiting cycle (which follows from
  Lemma~\ref{lem:optimal-cycles-ex}), this proves the lemma.

  So let $\beta$ be an optimal \T-visiting cycle and $\eta$ its bad
  segment (i.e., $\beta=\xi\cdot \eta\cdot \xi' $ for some finite
  paths $\xi,\xi'$). We denote $t=\eta(0)$ and $r=\last(\eta)$.  In
  the following we call every segment initiated in $t$ and ending in
  $r$ an $r$-$t$-segment. We also say that two paths are
  \T-equivalent, if both of them contain a state from \T{} or none of
  them does. We will construct an \rts $\eta'$ such that $\eta'$ is
  not bad, $\eta'$ is \T-equivalent ot $\eta$ and
  $\MC(\beta)=\MC(\beta')$, where $\beta'=\xi \cdot \eta'
  \cdot\xi'$. Then clearly $\beta'$ is an optimal \T-visiting cycle
  having at most $g-1$ bad segments.
 
  The construction proceeds in two steps. First we construct a compact
  \rts $\hat\eta$ of cost and length equal to $\cost(\eta)$ and
  $\len{\eta}$, respectively. Then we construct an \rts $\eta'$ with
  $\MC(\xi\cdot\eta'\cdot\xi')=\MC(\xi\cdot\hat\eta\cdot\xi')=\MC(\beta)$
  such that $\eta'$ is a compact segment with a characteristic, and we
  show that $\eta'$ must be optimal for all of its
  characteristics. During the construction we ensure that $\hat\eta$
  and $\eta'$ are \T-equivalent to
  $\eta$. %

  Note that from the optimality of $\beta$ it follows that every \rts
  of length equal to $\eta$ which is \T-equivalent to $\eta$ must have
  a cost greater or equal to $c(\eta)$. We will often use this fact in
  the proof.

  \textbf{Constructing $\hat\eta$ from $\eta$:} We employ a technique
  similar to the technique of decomposition into simple cycles. An
  \emph{$\eta$-decomposition } is a sequence
  $\dc=\alpha_0,\delta_0,k_0,\alpha_1,\delta_1,k_1,\dots,\alpha_{h-1},\delta_{h-1},k_{h-1},\alpha_h$
  such that
\begin{itemize}
\item For every $i$ the $\alpha_i$ is a finite path, $\delta_i$ is a
  simple cycle, and $k_i$ is a positive integer, and
\item $\alpha_0\delta_0^{k_0}\alpha_1\delta_1^{k_1}\cdots
  \alpha_{h-1}\delta_{h-1}^{k_{h-1}}\alpha_h$ is an \rts that is
  \T-equivalent to $\eta$ and whose cost and length are equal to
  $\cost(\eta)$ and $\len{\eta}$, respectively.
\end{itemize} 
A \emph{rank} of such an $\eta$-decomposition is the vector of natural
numbers $$\rank(\dc)\quad=\quad\Big(\sum_{i=0}^h
\len{\alpha_i},\,h,\,|\{i\mid k_i > |S|\}|\,\Big).$$

Now let
$\dc=\alpha_0,\delta_0,k_0,\alpha_1,\delta_1,k_1,\dots,\alpha_{h-1},\delta_{h-1},k_{h-1},\alpha_h$
be an $\eta$-decomposition with rank minimal w.r.t. the lexicographic
ordering (such an $\eta$-decomposition exists, since ranks are vectors
of natural numbers), and let
$\hat\eta=\alpha_0\delta_0^{k_0}\alpha_1\delta_1^{k_1}\cdots
\alpha_{h-1}\delta_{h-1}^{k_{h-1}}\alpha_h$ be the corresponding \rts
\T-equivalent to $\eta$, whose length and cost are equal to
$\len{\eta}$ and $\cost(\eta)$, respectively. We claim that the
following holds: for every $0\leq i \leq h$ we have $\len{\alpha_i} <
|S|$, $h\leq 2\cdot|S|$ and $|\{i\mid k_i > |S|\}| \leq 1$.
 From this it immediately follows that $\hat\eta=\gamma\cdot\delta^k\cdot \gamma'$, where $\delta$ is a cycle of length at most $|S|$, 
 and $\gamma,\gamma'$ are such that $\len{\gamma\cdot\gamma'}\leq 4|S|^3$. In particular, $\hat\eta$ is a compact segment (the fact that the paths $\gamma,\gamma'$ are shorter than required for the compactness will be used in the second part of the proof). 

 First let us assume, for the sake of contradiction, that for some $0
 \leq i \leq h$ it holds $\len{\alpha_i} \geq |S|$. Then $\alpha_i =
 \alpha'\delta'\alpha''$ for some simple cycle $\delta'$ of positive
 length and some (possibly empty) finite paths
 $\alpha',\alpha''$. Then the sequence
 $\dc'=\alpha_0,\delta_0,k_0,\dots,\delta_{i-1},{k_{i-1}},\alpha',\delta',1,\alpha'',\delta_{i},{k_{i}},\dots\alpha_h$
 is an $\eta$-decomposition such that
 $\rank(\dc')[1]=\rank(\dc)[1]-\len{\delta'}<\rank(\dc)[1]$, a
 contradiction with the choice of $\dc$.

 Now let us assume that $h>2\cdot|S|$. Then there are $0\leq i<j<h$
 such that $\len{\delta_{i}} = \len{\delta_{j}}$ and $\delta_i$ is
 \T-equivalent to $\delta_j$. It must be the case that
 $\cost(\delta_{i})=\cost(\delta_{j})$, otherwise, if
 e.g. $\cost(\delta_{i})<\cost(\delta_{j})$, then
 $\alpha_0\delta_0^{k_0}\cdots
 \delta_{i-1}^{k_{i-1}}\alpha_{i}\alpha_{i + 1}\cdots
 \alpha_{j}\delta_{j}^{k_{j}+k_{i}}\alpha_{j + 1}\cdots \alpha_h$
 would be an \rts \T-equivalent to $\theta$ whose length equals
 $\len{\eta}$ and whose cost is smaller than $\cost(\eta)$, a
 contradiction with the optimality of $\beta$. So the sequence $\dc'=
 \alpha_0,\dots,\delta_{i-1},k_{i-1},(\alpha_{i}\cdot\alpha_{i +
   1}),\delta_{i + 1},\dots,\alpha_{j},\delta_{j},k_{j} +
 k_{i},\alpha_{j + 1},\dots,\alpha_h$ is an $\eta$-decomposition with
 $\rank(dc')[1]=\rank(dc)[1]$ and $\rank(dc')[2]=\rank(dc)[2]-1$, a
 contradiction with the choice of $\dc$.

 Finally, let us assume that there are $0 \leq i<j< h$ such that $k_i
 > |S|$ and $k_j
 >|S|$. %
We distinguish three cases:  $\MC(\delta_i)>\MC(\delta_j)$, $\MC(\delta_i)<\MC(\delta_j)$ and $\MC(\delta_i)=\MC(\delta_j)$.

First assume that $\MC(\delta_i)>\MC(\delta_j)$ and let $a$ be the
greatest natural number such that $|S|\geq k_i -a\cdot
\len{\delta_j}\geq 1$ (clearly $a\geq 1$). We
have $$\MC(\delta_i)-\MC(\delta_j) =
\frac{\cost(\delta_i)\cdot\len{\delta_j} -
  \cost(\delta_j)\cdot\len{\delta_i}
}{{\len{\delta_i}\cdot{\len{\delta_j}}}} >
0%
, $$ from which it follows that $\cost(\delta_i)\cdot\len{\delta_j} -
\cost(\delta_j)\cdot\len{\delta_i}>0$. Now consider the path
$\pi=\alpha_0\delta_0^{k_0}\cdots \alpha_{i}\delta_i^{k_i -
  a\cdot\len{\delta_j}}\alpha_{i + 1}\cdots
\alpha_{j}\delta_{j}^{k_{j}+a\cdot\len{\delta_i}}\alpha_{j + 1}\cdots
\alpha_h$. Clearly $\len{\pi}=\len{\hat\eta}=\len{\eta}$ and $\pi$ is
an \rts \T-equivalent to
$\eta$. Moreover, $$\cost(\pi)=\cost(\hat\eta)-a\cdot(\cost(\delta_i)\cdot\len{\delta_j}
- \cost(\delta_j)\cdot\len{\delta_i})<\cost(\hat\eta)=\cost(\eta),$$ a
contradiction with the optimality of $\beta$.

The case $\MC(\delta_i)<\MC(\delta_j)$ is handled symmetrically, so it
remains to consider the case $\MC(\delta_i)=\MC(\delta_j)$. In this
case we clearly have $\cost(\delta_i)\cdot\len{\delta_j} -
\cost(\delta_j)\cdot\len{\delta_i}=0$ and thus the aforementioned path
$\pi$ is an \rts \T-equivalent to $\eta$ such that not only
$\len{\pi}=\len{\eta}$, but also $\cost(\pi)=\cost(\eta)$. It follows
that the corresponding sequence $\dc'=\alpha_0,\delta_0,{k_0},\cdots
\alpha_{i},\delta_i,{k_i - a\cdot\len{\delta_j}},\alpha_{i +
  1},\cdots,
\alpha_{j},\delta_{j},{k_{j}+a\cdot\len{\delta_i}},\alpha_{j +
  1},\cdots, \alpha_h$ is an $\eta$-decomposition such that
$\rank(\dc')[1]=\rank(\dc)[1]$, $\rank(\dc')[2]=\rank(\dc)[2]$ and
$\rank(\dc')[3]<\rank(\dc)[3]$, a contradiction with the choice of
$\dc$.

\textbf{Constructing $\eta'$ from $\hat\eta$:} We now have an \rts
$\hat\eta$ such that $\hat{\eta}$ is \T-equivalent to $\eta$,
$\len{\hat\eta}=\len{\eta}$ and $\cost(\hat\eta)=\cost(\eta)$ (and
thus also $\MC(\hat\eta)=\MC(\eta)$), and moreover
$\hat\eta=\gamma\cdot\delta^k\cdot\gamma'$ for some finite paths
$\gamma,\gamma'$ of combined length at most $4|S|^3$ and $\delta$
either a single vertex or a simple cycle. The compact segment
$\hat\eta$ may not have a characteristic for three reasons:
\begin{itemize} \item $c(\delta)=0$; 
\item $\delta$ contains a state from $T$ and $\gamma\cdot\gamma'$ not; 
\item $\delta$ is a cycle and $\gamma\cdot\delta^{k+1}\cdot\gamma'$ is also a $\Ca$-bounded path. 
\end{itemize}
The first two cases actually cannot happen. Indeed, if $c(\delta)=0$,
then $\delta$ is a zero-cost cycle (recall that every state on a
\T-visiting cycle is admissible, and $\delta(0)$ lies on a
\T-visiting cycle $\xi\cdot\hat\eta\cdot\xi'$), a contradiction with
the assumptions of Proposition~\ref{prop:algorithm}. In the second
case clearly $k\geq 1$ and
$\hat\eta=\gamma\cdot\delta\cdot\delta^{k-1}\cdot\gamma'$ is a compact
segment of a characteristic $(r,\delta(0),t,m,\len{\delta},1)$, where
$m=\len{\gamma\cdot\delta\cdot\gamma'}\leq 4|S|^3 +|S| \leq 5|S|^3$, a
contradiction with $\eta'$ not having a characteristic.

Now suppose that the first two cases do not occur and the third does. 
Then we have $k\geq 2$, since otherwise we would have
$\len{\hat\eta}\leq 4|S|^3 + |S|\leq 5|S|^3$ and ${\hat\eta}$ would
have a characteristic $(r,r,t,\len{\hat\eta},0,0)$. Now let $z\geq 1$
be the maximal number such that
$\eta'=\gamma\cdot\delta^{z}\cdot\gamma' $ is a $\Ca$-bounded
path. Clearly, $\eta'$ is a compact \rts \T-equivalent to $\hat\eta$
with a characteristic. We need to show that
$\MC(\xi\cdot\eta'\cdot\xi')=\MC(\xi\cdot\hat\eta\cdot\xi')$. From the
optimality of $\beta$ it follows that it suffices to show
$\MC(\xi\cdot\eta'\cdot\xi')\leq\MC(\xi\cdot\hat\eta\cdot\xi')$.

Assume, for the sake of contradiction, that $\MC(\xi\cdot\eta'\cdot\xi')>\MC(\xi\cdot\hat\eta\cdot\xi')$. Denote $C=\cost(\xi\cdot\gamma\cdot\gamma'\cdot\xi')$ and $D=\len{\xi\cdot\gamma\cdot\gamma'\cdot\xi'}$. Clearly
\begin{equation}
\label{eq:mc-iterace}
\MC(\xi\cdot\eta'\cdot\xi') = \frac{C+z\cdot \cost(\delta)}{D+z\cdot\len{\delta}} > \frac{C+k\cdot \cost(\delta)}{D+k\cdot\len{\delta}} = \MC(\xi\cdot\hat\eta\cdot\xi').
\end{equation}
Denote $z'\geq 1$ the number such that $z=k+z'$. From~\eqref{eq:mc-iterace} we gradually get
\begin{align*}
k\cdot C\cdot\len{\delta} + z\cdot D\cdot\cost(\delta) &>   k\cdot D\cdot\cost(\delta) +z\cdot C\cdot\len{\delta} \\
z'\cdot D\cdot\cost(\delta) &> z'\cdot C\cdot\len{\delta}\\
{}{D}\cdot\cost(\delta) &> {C}\cdot{\len{\delta}}\\
{(k-1)}\cdot{D}\cdot\cost(\delta) &> (k-1)\cdot{C}\cdot{\len{\delta}} \quad\text{(since $k\geq 2$)}\\
C \cdot\len{\delta} + k\cdot D\cdot\cost(\delta) &>   k\cdot C\cdot \len{\delta} +D\cdot\cost(\delta)\\
\frac{C+k\cdot \cost(\delta)}{D+k\cdot\len{\delta}} &> \frac{C+ \cost(\delta)}{D+\len{\delta}}.
\end{align*}
But then $(\xi\cdot\gamma\cdot\delta\cdot\gamma'\cdot\xi')$ is a \T-visiting cycle with $\MC(\xi\cdot\gamma\cdot\delta\cdot\gamma'\cdot\xi') = (C+\cost(\delta))/(D+\len{\delta}) < \MC(\xi\cdot\hat\eta\cdot\xi')=\MC(\beta)$, a contradiction with the optimality of $\beta$.

Now let $\chi=(r,q,t,m,n,b)$ be any characteristic of $\eta'$. We show that $\eta'$ is optimal for this characteristic. Assume, for the sake of contradiction, that it is not optimal for $\chi$. Two cases may happen:
\begin{itemize}
\item
There is an \rts $\alpha_0\cdot \alpha_1 $ of a %
characteristic $(r,q,t,m,0,b)$ such that $\cost(\alpha_0\cdot \alpha_1 )< \cost(\gamma\cdot\gamma')$. Then $\xi\cdot \alpha_0\cdot \delta^z\cdot \alpha_1 \cdot\xi'$ is a \T-visiting cycle with $\MC(\xi\cdot \alpha_0\cdot \delta^z\cdot \alpha_1 \cdot\xi')< \MC(\xi\cdot\eta'\cdot\xi')=\MC(\beta)$, a contradiction with the optimality of $\beta$.
\item
There is a cycle $\theta$ of length $n$ initiated in $q$ such that $\theta$ is either a segment or does not contain any reload state and $\cost(\theta)<\cost(\delta)$. Then $\xi\cdot \gamma\cdot \theta^z\cdot \gamma' \cdot\xi$ is again a \T-visiting cycle whose mean cost is smaller than $\MC(\beta)$, a contradiction.
\end{itemize}
\qed
\end{proof}

\subsection{A Proof of Lemma~\ref{lem:segments-algorithm}}
\label{sec:lem-segments-algorithm}

\begin{reflemma}{lem:segments-algorithm}
There is an algorithm which decides, for a given characteristic $\chi$, whether the set of all compact segments that have a characteristic $\chi$ is non-empty, and if the answer is \emph{yes}, it computes a tuple $(\gamma,\gamma',\delta,k)$ such that $\gamma\cdot\delta^k\cdot\gamma'$ is a compact segment optimal for $\chi$. The algorithm runs in polynomial time.
\end{reflemma}
\begin{proof}
Let $\chi=(r,q,t,m,n,b)$ be the input characteristic. From the definition of an optimal compact segment for $\chi$ it follows that we have to compute the following:
\begin{itemize}
\item If $n>0$ we have to compute a cycle $\delta^*$ of minimal cost among all cycles $\delta$ satisfying the following: $\len{\delta}=n$, $\delta(0)=q$, and $\delta$ is either a segment (if $q\in R$), or $\delta$ does not contain any reload state (if $q\not\in R$). In both cases we can use the algorithm of Lemma~\ref{lem:length-algorithm}, namely return the result of $\minpath(q,q,n,R)$ as the desired cycle $\delta^*$. If $\minpath(q,q,n,R)$ answers that the required path does not exist, we can immediately say that no compact segment has $\chi$ as its characteristic, i.e., we return ``no''.
\item If $b=0$, we have to compute a %
compact segment  
$\alpha^*$ of minimal cost among all 
compact segments that have a characteristic $(r,q,t,m,0,0)$.
Path $\alpha^*$ can be computed using the algorithm $\minpathr$ of Lemma~\ref{lem:length-reach-algorithm}, namely by making a call $\minpathr(r,t,m,R,\{q\})$. If the result of this call is a non-existence of the required path, we again immediately return ``no''.%
\item If $b=1$, we have to compute a %
compact segment  
$\alpha^*$ of minimal cost among all compact segments that have a characteristic $(r,q,t,m,0,1)$.
If $r$, $q$ or $t\in T$, we can proceed as in the previous case. Otherwise the path $\alpha^*$ can be computed as follows. For every $0\leq m' \leq m$ we compute these paths:
\begin{itemize}
\item $\alpha_{m',1}$ by calling $\minpathr(r,q,m',R,T)$,
\item $\alpha_{m-m',1}$ by calling $\minpath(q,t,m-m',R)$,
\item $\alpha_{m',2}$ by calling $\minpath(r,q,m',R)$,
\item $\alpha_{m-m',2}$ by calling $\minpathr(q,t,m-m',R,T)$.
\end{itemize}
If for all such $m'$ and all $i\in\{1,2\}$ one of the paths $\alpha_{m',i}$, $\alpha_{m-m',i}$ does not exist, we immediately return ``no''.%
Then we select $0\leq m'\leq m$ and $i\in\{1,2\}$ that minimizes $\cost(\alpha_{m',i}\cdot\alpha_{m-m',i})$ and we return $\alpha^*:=\alpha_{m',i}\cdot\alpha_{m-m',i}$. The correctness of this is clear, since every compact segment $\gamma\cdot\gamma'$ of a characteristic $(r,q,t,m,0,1)$ and of minimal cost has the property that  $\gamma$ or $\gamma'$ satisfies the conditions stated in Lemma~\ref{lem:length-reach-algorithm} (and it is of minimal cost among all paths satisfying these conditions), and both $\gamma$ and $\gamma'$ satisfy the conditions stated in Lemma~\ref{lem:length-algorithm}, with one of them having minimal cost among all paths satisfying this condition.
\end{itemize}
Now having the paths $\alpha^*$ (and $\delta^*$ if $n>0$) we write $\alpha^* = \gamma\cdot\gamma'$, where $\last(\gamma)=q$, and check whether $\cost(\alpha^*)\leq \Ca$ or $\cost(\gamma\cdot \delta*\cdot\gamma')\leq \Ca$, depending on whether $n=0$ or not. If this check fails, we return ``no.'' Otherwise, if $n=0$ we set $\delta^* = q$ and $k=1$, else we set $k=\lfloor(\Ca-\cost(\alpha^*))/\cost(\delta^*) \rfloor$ (this number exists since $\cost(\delta*)>0$ -- otherwise there would be a zero-cost cycle on an admissible state $q$, a contradiction with the assumptions of Proposition~\ref{prop:algorithm}). Clearly $\gamma\cdot(\delta^*)^k\gamma'$ is a compact segment optimal for $\chi$, so we return the tuple $(\gamma,\gamma',\delta^*,k)$ as the desired result.
\qed
\end{proof}

\subsection{A Proof of Lemma~\ref{col:cycle-with-segments}}

\begin{reflemma}{col:cycle-with-segments}
There is an optimal \T-visiting cycle $\beta$ such that every
segment of $\beta$ is of the form $\cpath(\chi)$ for some
characteristic $\chi$.
\end{reflemma}
\begin{proof}
  First note that if a compact segment $\eta$ contains a state from
  $T$ and at the same time it is optimal for some characteristic
  $\chi=(r,q,t,m,n,0)$, then $\eta$ also has a characteristic
  $\chi'=(r,q,t,m,n,1)$ (recall the last condition from the definition
  of a characteristic). Since every compact segment with a
  characteristic $\chi'$ also has a characteristic $\chi$, $\eta$ is
  optimal also for $\chi'$.

  Now among all optimal \T-visiting cycles whose all segments are
  compact and optimal for some characteristic (at least one exists due
  to Lemma~\ref{lemma:segment-candidates}) let $\beta$ be the one
  minimizing the number of segments that are not of the form
  $\cpath(\chi)$ for some $\chi$. Suppose that $\beta$ contains such a
  segment $\eta$ and denote $\chi$ a characteristic of $\eta$ for
  which $\eta$ is optimal. As mentioned above, if $\eta$ contains a
  state from $T$, then we may assume that the last component of $\chi$
  is 1. Write $\beta= \xi \cdot\eta \cdot\xi'$. By
  Lemma~\ref{lem:optimal-segments-char} we have
  $\len{\eta}=\len{\cpath(\chi)}$ and
  $\cost(\eta)=\cost(\cpath(\chi))$ and thus $\MC(\beta)=\MC(\beta')$
  where $\beta'= \xi \cdot\cpath(\chi) \cdot\xi'$.  Note that if
  $\eta$ contains a state from $T$ then, by our assumption, the last
  component of $\chi$ is 1, so $\cpath(\chi)$ also contains a state
  from
  $T$. %
  Thus, $\beta'$ is an optimal \T-visiting cycle containing smaller
  number of undesirable segments than $\beta$, a contradiction with
  the choice of $\beta$.
\qed
\end{proof}

\section{Proofs of Section~\ref{sec:limit}}

\subsection{A Proof of Lemma~\ref{lem:lim-val-char}}

\begin{reflemma}{lem:lim-val-char}
$\Val_\C(s)$ is finite iff there is a safe cycle, in which 
case  $\Val_\C(s)= \min\{\MC(\beta)\mid \beta \text{ is a safe cycle}  \}$.
Further, there is a finite $\Ca\in\Nset_0$ such that
$\Val^\Ca_\C(s)=\Val_\C(s)$ iff either $\Val_\C(s)=\infty$,
or there is a strongly safe cycle $\hat\beta$ such that
$\MC(\hat\beta)=\Val_\C(s)$. %
In such a case $\Val^\Ca_\C(s)=\Val_\C(s)$ for every $\Ca\geq
3\cdot|S|\cdot \cmax$, where $\cmax$ is the maximal cost of a
transition in $\C$.
\end{reflemma}
\begin{proof}

  Denote $\minsafe=\min\{\MC(\beta)\mid \beta \text{ is a safe cycle}
  \}$. We say that a simple cycle $\delta$ is a \emph{simple
    sub-cycle} of $\beta$ if $\beta=\xi\cdot\delta\cdot\xi'$ for some
  $\xi,\xi'$. We say that a set of simple cycles $D$ is a
  \emph{decomposition of $\beta$} (into simple cycles) if for every
  two distinct $\delta,\delta'\in D$ it holds that $\delta'$ is a
  simple sub-cycle of $\xi\cdot\xi'$, where
  $\beta=\xi\cdot\delta\cdot\xi'$. Note that for any decomposition $D$
  of $\beta$ it holds $\MC(\beta)\geq \min_{\delta\in D}\MC(\delta)$.

  Suppose that $\Val_\C(s)<\infty$. We show that a safe cycle exists
  and moreover, for any capacity $\Ca$ we have $\Val_\C^\Ca(s)\geq
  \minsafe$, from which it follows that $\Val_\C(s)\geq \minsafe$. If
  $\Val_\C^\Ca(s)=\infty$, the inequality is trivial. Otherwise the
  inequality follows from the first part of the following claim (its
  second part will be used later in the proof):
\begin{claim}
  If $\Val_\C^\Ca(s)<\infty$, then there is a safe cycle $\delta$ such
  that $\Val_\C^\Ca(s)\geq\MC(\delta)$. Moreover, if
  $\Val_\C^\Ca(s)=\Val_\C(s)$, then $\delta$ is strongly
  safe. %
\end{claim}

Indeed, from Lemma~\ref{prop:opt-runs-zero char}
and~\ref{prop:opt-runs-char} it follows that there is a path $\alpha$
and cycles $\beta,\gamma$ such that $\gamma$ contains an accepting
state and either $\cost(\beta)=0$ or $\beta(0)\in R$, and for a run
$\rho=\alpha\cdot\beta\cdot\gamma\cdot\beta^2\cdot\gamma\cdot\beta^4\cdots$
it holds $\MC(\beta)=\Val_\C^\Ca(\rho)=\Val_\C^\Ca(s)$. Note that
$\delta$ and $\gamma$ must be in the same strongly connected component
$C$ of $\C$. In particular, all states of %
$\beta$ are inter-reachable with an accepting state.

Now we consider two cases. Either $\cost(\beta\cdot\gamma)=0$, in which case $\beta\cdot\gamma$ has a sub-path $\delta$ which is a simple cycle of zero cost containing an accepting state -- i.e., $\delta$ is a strongly safe cycle and $\Val_\C^\Ca(s)=0=\MC(\delta)=\minsafe$.  

If $\cost(\beta\cdot\gamma)\neq 0$,
we distinguish two sub-cases. Either $\cost(\beta)=0$. Then $\cost(\gamma)\neq 0$ and it follows that $\beta\cdot\gamma$ contains a reload state. Then every simple sub-cycle $\delta$ of $\beta$ has $\delta(0)$ inter-reachable not only with an accepting state (as shown above), but also with a reload state and hence it is a strongly safe cycle (of zero cost). Here again $\Val_\C^\Ca(s)=0=\MC(\delta)=\minsafe$. 

The second sub-case is $\cost(\beta)\neq 0$. Then $\beta(0)\in R$ and thus all states on $\beta$ are inter-reachable with a reload state. Let $D$ be a decomposition of $\beta$ into simple cycles and $\delta\in D$ be the simple cycle of minimal mean cost. Clearly $\delta$ is a safe cycle and $\MC(\beta)\geq\MC(\delta)$, which finishes the proof of the first part of the claim. Now pick such a simple $\delta\in D$ which contains a reload state contained in $\beta$. Clearly this $\delta$ is strongly safe. We claim that either $\MC(\delta)\leq \Val_\C^\Ca(s)$ or $\Val_\C^\Ca(s)<\Val_\C(s)$, which finishes the proof of the second part of the claim. So suppose that $\MC(\delta)>\Val_\C^\Ca(s)=\MC(\beta)$ and write $\beta=\xi\cdot\delta\cdot\xi'$. Clearly $\MC(\xi\cdot\xi')<\MC(\delta)$ and thus for the cycle $\beta'=\xi\cdot\xi'\cdot \xi \cdot\delta\cdot\xi'$ it holds $\MC(\beta')<\MC(\beta)$. Moreover, the run $\rho'=\alpha\cdot\beta'\cdot\gamma\cdot{\beta'}^2\cdot\gamma\cdot{\beta'}^4\cdots$ is accepting and $\Ca'$-bounded for $\Ca' = \max\{\Ca,\cost{(\beta'\cdot\gamma)}\}$. Also note that $\Val_\C^{\Ca'}(\rho)=\MC(\beta')$ (this can be established via a straightforward computation identical to the one from the proof of Proposition~\ref{prop:opt-runs-char}). Thus, $\Val_{\C}^\Ca(s) =\MC(\beta)>\MC(\beta')\geq \Val_{\C}^{\Ca'}(s)\geq\Val_\C(s)$.

\vspace{2.2mm}
Conversely, suppose that there is a safe cycle and let $\beta$ be the one of minimal mean cost. We show that for every $1\geq \eps>0$ and every capacity $\Ca\geq \lceil(6|S|^2\cmax^2)/\eps\rceil$ there is a run $\rho_\eps$ such that $\Val_\C^\Ca(\rho_\eps)\leq \MC(\beta)+\eps$. From this it immediately follows that $\Val_\C(s)$ is finite, and in combination with the previous paragraph we get $\Val_\C(s)=\minsafe$. 

Let $\alpha$ be a shortest (w.r.t. the number of transitions) path from $s$ to $\beta(0)$. Note that $\cost(\alpha)\leq |S|\cdot\cmax$.
If $\cost(\beta)=0$ and $\beta$ contains an accepting state, then we can take $\rho_\eps:=\alpha\cdot\beta^\omega$, since for every $\Ca\geq|S|\cmax$ this is a $\Ca$-bounded accepting run with $\Val_\C^\Ca(\rho_\eps)=\MC(\beta)$. Otherwise let $\gamma_1$ be a shortest path from $\beta(0)$ to some accepting state $f$, $\gamma_2$ a shortest path from $f$ to some reload state $r$, and $\gamma_3$ a shortest path from $r$ to $\beta_0$, and put $\gamma=\gamma_1\cdot\gamma_2\cdot\gamma_3$. Note that $\cost(\gamma)\leq 3|S|\cmax$. Set $k=\lceil (3|S|\cmax)/\eps\rceil$. It easily follows that $\rho_\eps:=\alpha\cdot(\gamma\cdot\beta^k)^\omega$ is a $\Ca$-bounded accepting path for any $\Ca\geq \lceil(6|S|^2\cmax^2)/\eps\rceil$, since the consumption between two visits of the reload state $r$ on $\gamma$ is bounded by $\cost(\gamma)+k\cdot\cost(\beta)\leq 3|S|\cmax + 3|S|^2\cmax^2/\eps$. Let us compute $\MC(\rho_\eps)$. We have
\begin{align*}
\MC(\rho_\eps) = \MC(\gamma\cdot\beta^k) = \frac{\cost(\gamma) + k\cdot\cost(\beta)}{\len{\gamma} + k\cdot\len{\beta}} \leq \frac{\cost(\gamma)}{k\cdot\len{\beta}} + \MC(\beta).
\end{align*}
Now $\cost(\gamma)/(k\cdot\len{\beta})\leq \cost(\gamma)/k\leq 3|S|\cdot\cmax/k\leq \eps$ as required. 

Now suppose that there is a finite capacity $\Ca$ such that $\Val_{\C}^\Ca(s)=\Val_\C(s)$ ($=\minsafe$, as shown above). From the above claim it immediately follows that there is a strongly safe cycle $\delta$ such that $\MC(\delta)=\minsafe$. 

Conversely, suppose that $\hat\beta$ is a strongly safe cycle with $\MC(\hat\beta)=\Val_\C(s)$. Let $\alpha$ be again a shortest path from $s$ to $\beta(0)$. If $\cost(\hat\beta)=0$ and $\hat\beta$ contains an accepting state, then we again take $\rho=\alpha\cdot\hat{\beta}^{\omega}$ -- this is clearly an $(|S|\cdot\cmax)$-bounded run and $\Val_{\C}^{\Ca}(\rho)=0=\Val_\C(s)$ for any $\Ca \geq |S|\cdot\cmax$. Now suppose that $\cost(\hat\beta)=0$ and $\hat\beta(0)$ is inter-reachable with accepting state $f$ and with an accepting state $r$. Then there exists a cycle $\gamma$ initiated in $\hat\beta(0)$ that contains both $f$ and $r$. Then $\rho:=\alpha\cdot\hat\beta\cdot\gamma\cdot\hat{\beta}^2\cdot\gamma\cdot\hat{\beta}^4$ is a $(3|S|\cmax)$-bounded accepting run with $\Val_\C^{3|S|\cmax}=\MC(\hat\beta)=0=\Val_\C(s)$.

Finally, suppose that $\hat\beta$ contains a reload state and $\hat\beta(0)$ is inter-reachable with an accepting state. Let $\gamma$ be the cycle of minimal length among those initiated in $\hat\beta(0)$ that contain an accepting state. Then $\rho:=\alpha\cdot\hat\beta\cdot\gamma\cdot\hat{\beta}^2\cdot\gamma\cdot\hat{\beta}^4$ is again a $(3|S|\cmax)$-bounded accepting run, where the boundedness now comes from the fact that $\hat\beta$ contains a reload state and thus the end cost of $\hat{\beta}^{k}$ is at most the end cost of $\hat\beta$, which is at most $|S|\cdot\cmax$. Clearly, $\Val_{\C}^{3|S|\cmax}(\rho)=\MC(\beta)=\Val_{\C}(s).$

\qed 
\end{proof}

\subsection{A Proof of Lemma~\ref{lem:limit-algorithms}}

\begin{reflemma}{lem:limit-algorithms}
The existence of a safe (or strongly safe) cycle is decidable in polynomial
time. Further, if a safe (or strongly safe) cycle exists, then there
is a safe (or strongly safe) cycle $\beta$ computable in polynomial time
such that $\MC(\beta) \leq \MC(\beta')$ for every safe (or strongly safe)
cycle $\beta'$.
\end{reflemma}
\begin{proof}
To find a safe cycle $\beta$ of minimal mean cost (or to decide that no such cycle exists), we proceed as follows. First we compute the set $A$ of all states reachable from $s$. Then for every state $f \in F\cap A$ we use the polynomial-time algorithm of Lemma~\ref{lem:cycle-of-zero-cost} to find a cycle of zero cost initiated in $f$. If we find such a cycle $\delta$ for some $f\in F\cap A$, then clearly $\delta$ is a safe cycle of minimal (i.e., zero) mean cost. If we do not find such a cycle, we decompose $\C$ into its strongly connected components (SCCs), using, e.g., the Tarjan's algorithm. For every SCC $C$ that is reachable from $s$ we check, whether $C$ contains both a reload state and an accepting state. If no such component exists, we conclude that there is no safe cycle. Otherwise, for every SCC $C$ that contains both a reload state and an accepting state we compute a  cycle $\delta_\C$ of length at most $|S|$ of minimal mean cost in $C$, using standard polynomial time-algorithm for finding a cycle of the minimal mean cost (see, e.g.,~\cite{DBR:cost-time,DIG:cost-time}). \footnote{Note that given any cycle $\theta$ of minimal mean cost, we can easily extract from $\theta$ a simple cycle $\delta$ such that $\MC(\delta)=\MC(\theta)$. } Then clearly the cycle $\delta_{C^*}$ such that $\MC(\delta_{C^*})=\min \{\MC(\delta_C)\mid C \text{ is a SCC  containing a reload and an accepting state}\}$ is a safe cycle of minimal cost among all safe cycles in $\C$.

For strongly safe cycles we proceed in a similar way. First we check whether there is a cycle of zero cost containing a reachable accepting state using the same approach as in the previous paragraph. If we find such a cycle, than it is a strongly safe cycle of minimal mean cost. Otherwise, we again decompose $\C$ into SCCs and identify those SCCs reachable from $s$ that contain both a reload state and an accepting state. Let $X$ be the set of all such SCCs. If $X=\emptyset$, we immediately get that no strongly safe cycles exist. Otherwise for every SCC $C\in X$ we check, whether there is a cycle of zero cost in $C$, using again the algorithm from Lemma~\ref{lem:cycle-of-zero-cost} (we can also use the aforementioned algorithms for finding a cycle of minimal mean cost). If such a cycle exists for some $C\in X$, it is clearly a strongly safe cycle of minimal mean cost. Otherwise, we have to find a cycle $\delta$ of length at most $|S|$  in some $C\in X$ such that $\delta$ contains a reload state (and we of course need to find a cycle of minimal mean cost among all such cycles). To this end, for every $C\in X$ and every reload state $r\in C\cap R$ we construct a labelled graph $G_r=(V,\bbtran{},L,\ell)$, where $L\subset\Nset_0$ defined as follows:
\begin{itemize}
\item $V = C\times \{0,\dots,|S|\}$,
\item there is an edge $(q,i)\bbtran{a}(q',j)$ in $G_r$, here $0\leq i,j \leq |S|$, whenever $i=j+1$ and $q\tran{a}q'$ is a transition in $\C$,
\item for every $1 \leq i \leq |S|$ there is an edge $(r,i)\bbtran{0}(r,0)$,
\item there are no other edges.
\end{itemize}
Note that $G_r$ does not have to contain a cycle if $C$ does not contain a cycle, which may happen if $C$ contains a single state without a self-loop. If this is the case for all $C\in X$, we get that there are no strongly safe cycles. 
Otherwise note that every cycle in some $G_r$ contains the state $(r,0)$. Moreover, there is a natural many-to-one correspondence between the simple cycles in $G_r$ and cycles of length at most $|S|$ that are initiated in $r$ in $\C$, and this correspondence preserves the mean cost of the cycles. So in order to compute a cycle $\delta$ of minimal mean cost among all cycles that are initiated in a reload state of some $C\in X$, it suffices to compute, for every $C\in X$ and every $r \in C \cap R$ a simple cycle $\delta_r$ of minimal mean cost in $G_r$, using the standard algorithms mentioned above. If we then select $r^*$ such that $\MC(\delta_{r^*})=\min\{\MC(\delta_r) \mid C \in X, r\in C \cap R\}$, then from $\delta_{r^*}$ we can easily compute the corresponding cycle $\delta_{r^*}'$ in $\C$ that is a strongly safe cycle of minimal mean cost among all strongly safe cycles.

The correctness of the algorithm and its polynomial running time are immediate.
\qed
\end{proof}

\subsection{A Proof of Lemma~\ref{lem-limit-rate}}

\begin{reflemma}{lem-limit-rate}
 Let $\cmax$ be the maximal cost of a transition in $\C$. For every
   $\Ca> 4\cdot|S|\cdot\cmax$ we have that
   $$\Val_\C^\Ca(s)-\Val_\C(s)\leq
   \frac{3\cdot|S|\cdot\cmax}{\Ca-4\cdot|S|\cdot \cmax}.$$
\end{reflemma}
\begin{proof}
  The proof employs techniques very similar to those used in the proof
  of Lemma~\ref{lem:lim-val-char}.  If $\Val_\C(s)=\infty$, then the
  lemma is immediate. Otherwise by Lemma~\ref{lem:lim-val-char} there
  is a safe cycle $\beta$ such that $\MC(\beta)=\Val_{\C}(s)$. Let
  $\alpha$ be the path from $s$ to $\beta(0)$ of minimal length. If
  $\cost(\beta)=0$ and $\beta$ contains an accepting state, then for
  every $\Ca\geq |S|\cdot \cmax$ we have
  $\Val_\C^\Ca(s)=\Val_\C^{\Ca}(\alpha\cdot\beta^\omega)=\MC(\beta)=0$,
  and the lemma holds. It remains to consider the case when $\beta(0)$
  is inter-reachable with both a reload state and an accepting
  state. Then let $\gamma_1$ be a shortest (w.r.t. the number of
  transitions) path from $\beta(0)$ to some accepting state $f$,
  $\gamma_2$ a shortest path from $f$ to some reload state $r$, and
  $\gamma_3$ a shortest path from $r$ to $\beta_0$, and put
  $\gamma=\gamma_1\cdot\gamma_2\cdot\gamma_3$. Finally, put $k =
  \lfloor (\Ca-3\cdot|S|\cdot\cmax)/(\cost(\beta)) \rfloor$. Note that
  $k\geq 1$ since $\Ca\geq 4|S|\cmax$. Then
  $\rho:=\alpha\cdot(\gamma\cdot\beta^k)^\omega$ is a $\Ca$-bounded
  accepting run, since the consumption between two visits of the
  reload state $r$ on $\gamma$ is bounded by $3|S|\cmax + k\cdot
  \cost{(\beta)}\leq 3|S|\cmax + \Ca - 3|S|\cmax = \Ca$. Moreover,
\[
\Val_\C^\Ca(\rho)= \MC(\gamma\cdot\beta^\omega) = \frac{\cost(\gamma) + k\cdot\cost(\beta)}{\len{\gamma} + k\cdot\len{\beta}} \leq \frac{\cost(\gamma)}{k\cdot\len{\beta}} + \MC(\beta).
\]
Now 
$$\frac{\cost(\gamma)}{k\cdot\len{\gamma}}\leq \frac{\cost(\gamma)}{\left(\frac{\Ca-3|S|\cmax}{\cost(\beta)}-1\right)\cdot\cost(\beta)}\leq \frac{3|S|\cmax}{\Ca-3|S|\cmax -\cost(\beta)}\leq\frac{3|S|\cmax}{\Ca-4|S|\cmax} $$ as required.
\qed
\end{proof}

%% file: main.bbl
\begin{thebibliography}{10}

\bibitem{FST-TCS2010}
{\em {Proceedings of {FST\&TCS 2010}}}, volume~8 of {\em Leibniz International
  Proceedings in Informatics}. Schloss Dagstuhl--Leibniz-Zentrum f{\"{u}}r
  Informatik, 2010.

\bibitem{Icalp2010-II}
{\em {Proceedings of ICALP 2010, Part II}}, volume 6199 of {\em Lecture Notes
  in Computer Science}. Springer, 2010.

\bibitem{BKSV:Solvency-games}
N.~Berger, N.~Kapur, L.J.{} Schulman, and V.~Vazirani.
\newblock {Solvency Games}.
\newblock In {\em {Proceedings of {FST\&TCS 2008}}}, volume~2 of {\em Leibniz
  International Proceedings in Informatics}, pages 61--72. Schloss
  Dagstuhl--Leibniz-Zentrum f{\"{u}}r Informatik, 2008.

\bibitem{BFLMS:weighted-automata-inf-runs}
P.~Bouyer, U.~Fahrenberg, K.~Larsen, N.~Markey, and J.~Srba.
\newblock {Infinite Runs in Weighted Timed Automata with Energy Constraints}.
\newblock In {\em {Proceedings of {FORMATS 2008}}}, volume 5215 of {\em Lecture
  Notes in Computer Science}, pages 33--47. Springer, 2008.

\bibitem{BBE:OC-games}
T.~Br{\'a}zdil, V.~Bro\v{z}ek, and K.~Etessami.
\newblock {One-Counter Stochastic Games}.
\newblock In {\em {Proceedings of {FST\&TCS 2010}}\/} \cite{FST-TCS2010}, pages
  108--119.

\bibitem{BBEKW:OC-MDP}
T.~Br{\'a}zdil, V.~Bro\v{z}ek, K.~Etessami, A.~Ku\v{c}era, and D.~Wojtczak.
\newblock {One-Counter {Markov} Decision Processes}.
\newblock In {\em {Proceedings of SODA 2010}}, pages 863--874. SIAM, 2010.

\bibitem{BCKN:consumption-games}
T.~Br{\'a}zdil, K.~Chatterjee, A.~Ku\v{c}era, and P.~Novotn{\'y}.
\newblock {Efficient Controller Synthesis for Consumption Games with Multiple
  Resource Types}.
\newblock In {\em {Proceedings of CAV 2012}}, volume 7358 of {\em Lecture Notes
  in Computer Science}, pages 23--38. Springer, 2012.

\bibitem{BJK:eVASS-games}
T.~Br{\'a}zdil, P.~Jan\v{c}ar, and A.~Ku\v{c}era.
\newblock {Reachability Games on Extended Vector Addition Systems with States}.
\newblock In {\em {Proceedings of ICALP 2010, Part II}\/} \cite{Icalp2010-II},
  pages 478--489.

\bibitem{BKNW:OC-MDP-term-time}
T.~Br{\'a}zdil, A.~Ku\v{c}era, P.~Novotn{\'y}, and D.~Wojtczak.
\newblock {Minimizing Expected Termination Time in One-Counter {Markov}
  Decision Processes}.
\newblock In {\em {Proceedings of ICALP 2012, Part II}}, volume 7392 of {\em
  Lecture Notes in Computer Science}, pages 141--152. Springer, 2012.

\bibitem{CHD:energy-games}
K.~Chatterjee and L.~Doyen.
\newblock {Energy Parity Games}.
\newblock In {\em {Proceedings of ICALP 2010, Part II}\/} \cite{Icalp2010-II},
  pages 599--610.

\bibitem{CHD:energy-MDPs}
K.~Chatterjee and L.~Doyen.
\newblock {Energy and Mean-Payoff Parity {Markov} Decision Processes}.
\newblock In {\em {Proceedings of MFCS 2011}}, volume 6907 of {\em Lecture
  Notes in Computer Science}, pages 206--218. Springer, 2011.

\bibitem{CHDHR:energy-mean-payoff}
K.~Chatterjee, L.~Doyen, T.~Henzinger, and J.-F.{} Raskin.
\newblock {Generalized Mean-payoff and Energy Games}.
\newblock In {\em {Proceedings of {FST\&TCS 2010}}\/} \cite{FST-TCS2010}, pages
  505--516.

\bibitem{CHKN:energy-games-polynomial}
K.~Chatterjee, M.~Henzinger, S.~Krinninger, and D.~Nanongkai.
\newblock {Polynomial-Time Algorithms for Energy Games with Special Weight
  Structures}.
\newblock In {\em {Proceedings of ESA 2012}}, volume 7501 of {\em Lecture Notes
  in Computer Science}, pages 301--312. Springer, 2012.

\bibitem{CHJ:MP-parity-games}
K.~Chatterjee, T.~Henzinger, and M.~Jurdzi{\'n}ski.
\newblock {Mean-Payoff Parity Games}.
\newblock In {\em {Proceedings of LICS 2005}}, pages 178--187. IEEE Computer
  Society Press, 2005.

\bibitem{DBR:cost-time}
B.~Dantzig, W.~Blattner, and M.~R. Rao.
\newblock {Finding a cycle in a graph with minimum cost to times ratio with
  applications to a ship routing problem}.
\newblock In P.~Rosenstiehl, editor, {\em {Theory of Graphs}}, pages 77--84.
  Gordon and Breach, 1967.

\bibitem{DIG:cost-time}
A.~Dasdan, S.S. Irani, and R.K. Gupta.
\newblock {Efficient algorithms for optimum cycle mean and optimum cost to time
  ratio problems}.
\newblock In {\em {Design Automation Conference, 1999. Proceedings. 36th}},
  pages 37--42, 1999.

\bibitem{FJLS:multi-energy-games}
U.~Fahrenberg, L.~Juhl, K.~Larsen, and J.~Srba.
\newblock {Energy Games in Multiweighted Automata}.
\newblock In {\em {Proceedings of the 8th International Colloquium on
  Theoretical Aspects of Computing {(ICTAC'11)}}}, volume 6916 of {\em Lecture
  Notes in Computer Science}, pages 95--115. Springer, 2011.

\bibitem{Kucera:multicounter-games}
A.~Ku\v{c}era.
\newblock {Playing Games with Counter Automata}.
\newblock In {\em {Reachability Problems}}, volume 7550 of {\em Lecture Notes
  in Computer Science}, pages 29--41. Springer, 2012.

\end{thebibliography}
